\documentclass[a4paper,UKenglish]{lipics-v2018}

\usepackage{microtype}

\title{It Is Easy to Be Wise After the Event: Communicating Finite-State Machines Capture First-Order Logic with ``Happened Before''}

\titlerunning{It Is Easy to Be Wise After the Event}

\author{Benedikt Bollig}{LSV, CNRS \& ENS Paris-Saclay, Universit{\'e} Paris-Saclay, Cachan, France}{bollig@lsv.fr}{}{}

\author{Marie Fortin}{LSV, CNRS \& ENS Paris-Saclay, Universit{\'e} Paris-Saclay, Cachan, France}{fortin@lsv.fr}{}{}

\author{Paul Gastin}{LSV, CNRS \& ENS Paris-Saclay, Universit{\'e} Paris-Saclay, Cachan, France}{gastin@lsv.fr}{}{}

\authorrunning{B. Bollig, M. Fortin, and P. Gastin}

\Copyright{B. Bollig, M. Fortin, and P. Gastin}

\subjclass{Theory of computation $\rightarrow$ Concurrency, Theory of computation $\rightarrow$ Logic and verification}


\category{}

\relatedversion{This is the full version of a CONCUR'18 paper, which is available at the following link: \url{http://dx.doi.org/10.4230/LIPIcs.CONCUR.2018.7}}

\funding{Partly supported by ANR FREDDA (ANR-17-CE40-0013) and ReLaX, 
UMI2000 (CNRS, ENS Paris-Saclay, Univ.\ Bordeaux, CMI, IMSc).}

\nolinenumbers

\hideLIPIcs

\EventEditors{Sven Schewe and Lijun Zhang}
\EventNoEds{2}
\EventLongTitle{29th International Conference on Concurrency Theory (CONCUR 2018)}
\EventShortTitle{CONCUR 2018}
\EventAcronym{CONCUR}
\EventYear{2018}
\EventDate{September 4--7, 2018}
\EventLocation{Beijing, China}
\EventLogo{}
\SeriesVolume{118}
\ArticleNo{7} % “New number” (=<article-no>) goes here!

\makeatletter
\renewcommand{\paragraph}{\@startsection{paragraph}{6}{\z@}{2ex}{-0.7em}{\normalsize\bf}}
\makeatother

\usepackage{amsmath,amssymb}
\usepackage{stmaryrd,multirow}
\usepackage{fancyhdr}
\usepackage{xspace}
\usepackage{mathtools}
\usepackage{cite}
\usepackage{tikz}
\usetikzlibrary{trees,arrows,shapes,snakes,fit,shadows,decorations.text,decorations.pathmorphing,decorations.pathreplacing,calc,automata,positioning,patterns}
\usepackage{enumitem}

\usepackage{ifpdf}

\colorlet{colj}{gray}
\definecolor{coll}{RGB}{31,0,149}
\definecolor{colk}{RGB}{0,152,83}
\definecolor{mygreen}{RGB}{0,152,83}

\theoremstyle{plain}
\newtheorem{proposition}[theorem]{Proposition}
\newtheorem{claim}[theorem]{Claim}
\newtheorem{fact}[theorem]{Fact}

\usepackage{amsmath,amssymb}
\usepackage{stmaryrd}
\usepackage{xspace}
\usepackage{mathtools}
\usepackage{tikz}
\usetikzlibrary{trees,arrows,shapes,snakes,fit,shadows,decorations.text,decorations.pathmorphing,decorations.pathreplacing,calc,automata,positioning,patterns,arrows.meta}

\newcommand\Nat{\mathbb{N}}
\newcommand{\df}{:=}
\newcommand\A{\mathcal{A}}
\newcommand\B{\mathcal{B}}

\newcommand\MSCs[2]{\mathbb{MSC}(#1,#2)}
\newcommand\ebMSCs[3]{\mathbb{MSC}_{\exists #3}(#1,#2)}

\newcommand\prel{\rightarrow}
\newcommand\mrel{\lhd}
\newcommand{\ploc}{\mathit{loc}}
\newcommand{\proc}{\mathsf{proc}}
\newcommand{\Ch}{\mathit{Ch}}
\newcommand{\N}{\mathbb{N}}
\newcommand{\col}{\mathit{col}}

\newcommand{\PDL}{\sfPDLall}
\newcommand{\sfPDL}{\textup{PDL}_{\mathsf{sf}}}
\newcommand{\PDLm}{\sfPDLm}

\newcommand{\MSO}{\textup{MSO}}
\newcommand{\EMSO}{\textup{EMSO}}
\newcommand{\FO}{\textup{FO}}
\newcommand{\FOle}{\FO[\prel,\mrel,\le]}

\newcommand{\rightg}[1]{\xrightarrow{#1}}
\newcommand{\leftg}[1]{\xleftarrow{#1}}
\newcommand{\leproc}{\leq_\proc}
\newcommand{\lessproc}{<_\proc}
\newcommand{\ltproc}{<_\proc}
\newcommand{\righta}{\xrightarrow{\ast}}

\newcommand{\rightp}{\xrightarrow{+}}
\newcommand{\leftp}{\xleftarrow{+}}
\newcommand{\leftmove}{\leftarrow}
\newcommand{\jump}[2]{\mathsf{jump}_{#1,#2}}
\newcommand{\test}[1]{\{#1\}?}
\newcommand{\Loop}[1]{\mathsf{Loop}(#1)}
\newcommand{\existsp}[2]{\mathop{\langle #1 \rangle} {#2}}
\newcommand{\existsptrue}[1]{\mathop{\langle #1 \rangle}}
\newcommand{\pic}[1]{{#1}^{\mathsf{c}}}
\newcommand{\True}{\mathit{true}}
\newcommand{\False}{\mathit{false}}

\newcommand{\compl}{\mathsf{c}}
\newcommand{\Loopname}{\mathsf{Loop}}

\newcommand{\sfPDLall}{\sfPDL}

\newcommand{\sfPDLm}{\sfPDL[\Loopname]}
\newcommand{\E}{\mathop{\mathsf{E}\vphantom{a}}\nolimits}

\newcommand{\sem}[2]{\llbracket {#2} \rrbracket_{#1}}
\newcommand{\semM}[1]{\llbracket {#1} \rrbracket}

\newcommand{\Comp}[1]{\mathsf{Comp}(#1)}
\newcommand{\id}{\mathsf{id}}
\newcommand{\Procs}{P}

\newcommand{\pdlsentence}{\xi}
\newcommand{\pifo}[1]{\widetilde{#1}}
\newcommand{\phifo}[1]{\widetilde{#1}}

\newcommand{\minpi}[1]{{\mathsf{min}~} {#1}}
\newcommand{\maxpi}[1]{{\mathsf{max}~} {#1}}

\newcommand{\minpie}[2]{\min \semM {#1} (#2)}
\newcommand{\maxpie}[2]{\max \semM {#1} (#2)}

\newcommand{\Loopop}{\mathsf{Loop}}
\newcommand{\Lt}[1]{\llbracket #1 \rrbracket}
\newcommand{\Aphi}[1]{\A_{#1}}
\newcommand{\Mphi}[2]{{#1}_{#2}}

\newcommand{\ncolone}{
\begin{tikzpicture}
\filldraw[fill=white,draw=black] circle (2.8pt);
\end{tikzpicture}
}
\newcommand{\ncoltwo}{
\begin{tikzpicture}
\filldraw[fill=gray!90,draw=black] circle (2.8pt);
\end{tikzpicture}
}
\newcommand{\ycolone}{
\begin{tikzpicture}
\filldraw[fill=white,draw=black] rectangle (5pt,5pt);
\end{tikzpicture}
}
\newcommand{\ycoltwo}{
\begin{tikzpicture}
\filldraw[fill=gray!90,draw=black] rectangle (5pt,5pt);
\end{tikzpicture}
}

\newcommand\init{\iota}
\newcommand\Acc{\mathit{Acc}}
\newcommand\Msg{\mathit{Msg}}
\newcommand\msg{m}
\newcommand{\act}{\alpha}
\newcommand{\loc}[1]{\langle#1\rangle}
\newcommand{\send}[3]{\langle#1,!_{#3}#2\rangle}
\newcommand{\rec}[3]{\langle#1,?_{#3}#2\rangle}
\newcommand{\Act}[2]{\mathit{Act}_{#1}(#2)}
\newcommand\source{\mathit{source}}
\newcommand\target{\mathit{target}}
\newcommand\tmsg{\mathit{msg}}
\newcommand\tlabel{\mathit{label}}
\newcommand\receiver{\mathit{receiver}}
\newcommand\sender{\mathit{sender}}

\newcommand{\fophi}{\Phi}

\newcommand{\Class}{\mathcal{F}}
\newcommand{\Sign}{R}
\newcommand{\msclang}{\mathbb{L}}
\newcommand{\Lang}{\mathcal{L}}
\newcommand{\LangCFM}{\Lang(\textup{CFM})}
\newcommand{\Var}{\mathcal{V}_{\mathsf{event}}}
\newcommand{\VAR}{\mathcal{V}_{\mathsf{set}}}
\newcommand{\Free}{\mathsf{Free}}

\tikzstyle{dot} = [circle, fill, inner sep=0, minimum size = 4pt]

\pgfdeclarelayer{background}
\pgfdeclarelayer{foreground}
\pgfsetlayers{background,main,foreground}

\begin{document}

\maketitle

\begin{abstract}
  Message sequence charts (MSCs) naturally arise as executions of communicating
  finite-state machines (CFMs), in which finite-state processes exchange
  messages through unbounded FIFO channels.  We study the first-order logic of
  MSCs, featuring Lamport's happened-before relation.  We introduce a star-free
  version of propositional dynamic logic (PDL) with loop and converse.  Our main
  results state that (i) every first-order sentence can be transformed into an
  equivalent star-free PDL sentence (and conversely), and (ii) every star-free
  PDL sentence can be translated into an equivalent CFM. This answers an open
  question and settles the exact relation between CFMs and fragments of monadic
  second-order logic.  As a byproduct, we show that first-order logic over MSCs
  has the three-variable property.
\end{abstract}

\section{Introduction}

First-order (FO) logic can be considered, in many ways, a reference
specification language.  It plays a key role in automated theorem proving and
formal verification.  In particular, FO logic over finite or infinite words is
central in the verification of reactive systems.  When a word is understood as a
total order that reflects a chronological succession of events, it represents an
execution of a sequential system.  Apart from being a natural concept in itself,
FO logic over words enjoys manifold characterizations.  It defines exactly 
the star-free languages and coincides with recognizability by aperiodic monoids
or natural subclasses of finite (B{\"u}chi, respectively) automata (cf.\
\cite{DiGa08Thomas,Tho97handbook} for overviews).  Moreover, linear-time
temporal logics are usually measured against their expressive power with respect
to FO logic.  For example, LTL is considered the yardstick temporal logic not
least due to Kamp's famous theorem, stating that LTL and FO logic are
expressively equivalent \cite{Kamp68}.

While FO logic on words is well understood, a lot remains to be said once
concurrency enters into the picture.  When several processes communicate
through, say, unbounded first-in first-out (FIFO) channels, 
events are only partially ordered and a behavior, which is referred to
as a \emph{message sequence chart (MSC)}, reflects Lamport's happened-before
relation: an event $e$ happens before an event $f$ if, and only if, there is a
``message flow'' path from $e$ to $f$ \cite{Lamport78}.  \emph{Communicating
finite-state machines} (CFMs) \cite{Brand1983} are to MSCs what finite automata
are to words: a canonical model of finite-state processes that communicate
through unbounded FIFO channels.  Therefore, the FO logic of MSCs can be
considered a canonical specification language for such systems.  Unfortunately,
its study turned out to be difficult, since algebraic and automata-theoretic
approaches that work for words, trees, or Mazurkiewicz traces do not carry over.
In particular, until now, the following central problem remained open:

\begin{center}
  \parbox{0.7\textwidth}{ \it Can every first-order sentence be transformed into an equivalent communicating finite-state machine, without any channel bounds?  }
\end{center}

Partial answers were given for CFMs with bounded channel capacity
\cite{HenriksenJournal,Kuske01,GKM06} and for fragments of FO that restrict the
logic to bounded-degree predicates \cite{BolligJournal} or to two variables
\cite{BFG-stacs18}.

In this paper, we answer the general question positively.  To do so, we
make a detour through a variant of propositional dynamic logic (PDL)
with loop and converse \cite{FisL79,Streett81}.
Actually, we introduce \emph{star-free} PDL, which
serves as an interface between FO logic and CFMs.
That is, there are two main tasks to accomplish:
\begin{itemize}
\item[(i)] Translate every FO sentence into a star-free PDL sentence.

\item[(ii)] Translate every star-free PDL sentence into a CFM.
\end{itemize}

Both parts constitute results of own interest.
In particular, step (i) implies that, over MSCs, FO logic has the three-variable property, i.e.,
every FO sentence over MSCs can be rewritten into one that uses only three different
variable names.
Note that this is already interesting in the special
case of words, where it follows from Kamp's theorem \cite{Kamp68}.
It is also noteworthy that star-free PDL is a \emph{two-dimensional} temporal logic in the sense of
Gabbay et al.~\cite{Gabbay1981,gabbay1994temporal}. 
Since every star-free PDL sentence is equivalent to some FO sentence, we actually provide a (higher-dimensional) temporal logic over MSCs that is expressively
complete for FO logic.%
\footnote{It is open whether there is an equivalent one-dimensional one.}
While step (i) is based on purely logical considerations, step (ii)
builds on new automata constructions that allow us to cope with the loop operator of PDL.

Combining (i) and (ii) yields the translation from FO logic to CFMs.
It follows that CFMs are expressively equivalent to \emph{existential} MSO logic.
Moreover, we can derive self-contained proofs of several results on
channel-bounded CFMs whose original proofs refer to involved constructions for Mazurkiewicz traces
(cf.\ Section~\ref{sec:concl}).

\subparagraph{Related Work.}

Let us give a brief account of what was already known on the relation between logic and CFMs.
In the 60s, B{\"u}chi, Elgot, and Trakhtenbrot proved that 
finite automata over words are expressively equivalent to monadic second-order logic \cite{Buechi:60,Elgot1961,Trakhtenbrot62}.
Note that finite automata correspond to the special case of CFMs with a single process.

This classical result has been generalized to CFMs with bounded channels:
Over \emph{universally} bounded MSCs (where all possible linear extensions meet a given channel bound),
deterministic CFMs are expressively equivalent to MSO logic \cite{HenriksenJournal,Kuske01}.
Over \emph{existentially} bounded MSCs (some linear extension meets the channel bound),
CFMs are still expressively equivalent to MSO logic \cite{GKM06}, but inherently nondeterministic
\cite{GKM07}.  The proofs of these characterizations reduce message-passing
systems to finite-state shared-memory systems so that deep results from
Mazurkiewicz trace theory \cite{DiekertRozenberg95} can be applied.

This generic approach is no longer applicable when the restriction on the channel capacity is dropped. Actually, in general, CFMs do not capture MSO logic \cite{BolligJournal}. On the other hand, they are expressively equivalent to existential MSO logic when we discard the happened-before relation \cite{BolligJournal} or when restricting to two first-order variables \cite{BFG-stacs18}. Both results rely on normal forms of FO logic, due to Hanf \cite{Hanf1965} and Scott \cite{GradelO99}, respectively.
However, MSCs with the happened-before relation are structures of \emph{unbounded} degree (while Hanf's normal form requires structures of bounded degree), and we consider FO logic with \emph{arbitrarily} many variables (while Scott's normal form only applies to two-variable logic).
That is, neither approach is applicable in our case.

Finally, there exists a translation of a loop-free PDL into CFMs \cite{BKM-lmcs10}.
As our star-free PDL has a loop operator, we cannot exploit \cite{BKM-lmcs10} either.

\subparagraph{Outline.}
In Section~\ref{sec:prel}, we recall basic notions such as MSCs, FO logic, and CFMs.
Moreover, we state one of our main results: the translation of FO formulas into CFMs.
Section~\ref{sec:pdl} presents star-free PDL and proves that it captures FO
logic.  In Section~\ref{sec:pdl-cfm}, we establish the translation of star-free
PDL into CFMs.  We conclude in Section~\ref{sec:concl} mentioning applications
of our results.  

% ------------------------------------------------------------
% ------------------------------------------------------------
% ------------------------------------------------------------

\section{Preliminaries}\label{sec:prel}

We consider message-passing systems in which processes communicate through unbounded FIFO channels.
We fix a nonempty finite set of \emph{processes} $\Procs$ and a nonempty finite set of \emph{labels} $\Sigma$.
For all $p,q \in \Procs$ such that $p \neq q$, there is a channel $(p,q)$ that allows $p$ to send messages to $q$. The set of channels is denoted $\Ch$.

In the following, we define message sequence charts, which represent executions
of a message-passing system, and logics to reason about them.  Then, we recall
the definition of communicating finite-state machines and state one of our main
results.

\subsection{Message Sequence Charts}

A \emph{message sequence chart (MSC)} (over $\Procs$ and $\Sigma$) is a graph $M=(E,\prel,\mrel,\ploc,\lambda)$ with nonempty finite set of nodes $E$, edge relations ${\prel},{\mrel} \subseteq E \times E$, and node-labeling functions $\ploc\colon E \to \Procs$ and $\lambda\colon E \to \Sigma$. An example MSC is depicted in Figure~\ref{fig:msc}.
A node $e \in E$ is an \emph{event} that is executed by process $\ploc(e) \in \Procs$.
In particular, $E_p \df \{e \in E \mid \ploc(e) = p\}$ is the set of events located on $p$.
The label $\lambda(e) \in \Sigma$ may provide more information about $e$ such as the message that is sent/received at $e$ or ``enter critical section'' or ``output some value''.

Edges describe causal dependencies between events:
\begin{itemize}
\item The relation $\prel$ contains \emph{process edges}.  They connect
successive events executed by the same process.  That is, we actually have
${\to} \subseteq \bigcup_{p \in \Procs} (E_p \times E_p)$.  Every process $p$ is
sequential so that ${\to} \cap (E_p \times E_p)$ must be the direct-successor
relation of some total order on $E_p$.  We let ${\leproc} \df {\to^\ast}$ and
${\lessproc} \df {\to^+}$.

\item The relation $\mrel$ contains \emph{message edges}. If $e \mrel f$, then $e$ is a \emph{send event} and $f$ is the corresponding \emph{receive event}.
In particular, $(\ploc(e),\ploc(f)) \in \Ch$.
Each event is part of at most one message edge. An event that is neither a send nor a receive event is called \emph{internal}. Moreover, for all $(p,q) \in \Ch$ and $(e,f), (e',f') \in {\mrel} \cap (E_p
\times E_q)$, we have $e \leproc e'$ iff $f \leproc f'$ (which guarantees a FIFO behavior).
\end{itemize}
We require that ${\prel} \cup {\mrel}$ be acyclic (intuitively, messages cannot travel backwards in time). The associated partial order is denoted ${\le} \df ({\prel} \cup {\mrel})^*$ with strict part ${<}=({\prel} \cup {\mrel})^+$.
We do not distinguish isomorphic MSCs.
Let $\MSCs{\Procs}{\Sigma}$ denote the set of MSCs over $\Procs$ and $\Sigma$.

Actually, MSCs are very similar to the space-time diagrams from Lamport's seminal paper \cite{Lamport78}, and $\le$ is commonly referred to as the \emph{happened-before relation}.

It is worth noting that, when $\Procs$ is a singleton, an MSC with events $e_1
\to e_2 \to \ldots \to e_n$ can be identified with the word
$\lambda(e_1)\lambda(e_2) \ldots \lambda(e_n) \in \Sigma^\ast$.

\newcommand{\nbullet}{\resizebox{!}{1.4ex}{$\diamond$}}

\newcommand{\abullet}{
\begin{tikzpicture}
\filldraw[fill=white,draw=black] circle (2.8pt);
\end{tikzpicture}
}

\newcommand{\bbullet}{
\begin{tikzpicture}
\filldraw[fill=white,draw=black] rectangle (5pt,5pt);
\end{tikzpicture}
}

\tikzstyle{acirc} = [draw, fill=white, circle, inner sep=0, minimum size=0.25cm, draw=black]
\tikzstyle{bcirc} = [draw, fill=white, rectangle, inner sep=0, minimum size=0.2cm, draw=black]
\tikzstyle{ncirc} = [draw, fill=white, diamond, inner sep=0, minimum size=0.3cm]

\begin{example}
  Consider the MSC from Figure~\ref{fig:msc} over $\Procs=\{p_1,p_2,p_3\}$ and
  $\Sigma = \{\bbullet,\abullet,\nbullet\}$.  We have, for instance, $E_{p_1} =
  \{e_0,\ldots,e_7\}$.  The process relation is given by $e_i \to e_{i+1}$, $f_i
  \to f_{i+1}$, and $g_i \to g_{i+1}$ for all $i \in \{0,\ldots,6\}$.
  Concerning the message relation, we have $e_1 \mrel f_0$, $e_4 \mrel g_5$,
  etc.  Moreover, $e_2 \le f_3$, but neither $e_2 \le f_1$ nor $f_1 \le
  e_2$.
\end{example}

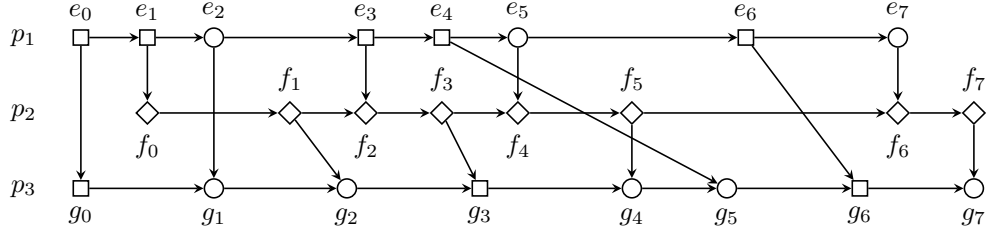
\begin{figure}[t]
\centering
    \begin{tikzpicture}[semithick,>=stealth]

     \node[bcirc,label=above:$e_0$] (e0) at (0.25,2) {};
     \node[bcirc,label=below:$g_0$] (f0) at (0.25,0) {};
     \draw[->] (e0) -- (f0);

     \node[bcirc,label=above:$e_1$] (e1) at (1.125,2) {};
     \node[ncirc,label=below:$f_0$] (g0) at (1.125,1) {};
     \draw[->] (e1) -- (g0);

     \node[acirc,label=above:$e_2$] (e2) at (2,2) {};
     \node[acirc,label=below:$g_1$] (f1) at (2,0) {};
     \draw[->] (e2) -- (f1);

     \node[ncirc,label=above:$f_1$] (g1) at (3,1) {};
     \node[acirc,label=below:$g_2$] (f2) at (3.75,0) {};
     \draw[->] (g1) -- (f2);

     \node[bcirc,label=above:$e_3$] (e3) at (4,2) {};
     \node[ncirc,label=below:$f_2$] (g2) at (4,1) {};
     \draw[->] (e3) -- (g2);

     \node[ncirc,label=above:$f_3$] (g3) at (5,1) {};
     \node[bcirc,label=below:$g_3$] (f3) at (5.5,0) {};
     \draw[->] (g3) -- (f3);
     
     \node[bcirc,label=above:$e_4$] (e4) at (5,2) {};
     \node[acirc,label=below:$g_5$] (f5) at (8.75,0) {};
     \draw[->] (e4) -- (f5);

     \node[acirc,label=above:$e_5$] (e5) at (6,2) {};
     \node[ncirc,label=below:$f_4$] (g4) at (6,1) {};
     \draw[->] (e5) -- (g4);

     \node[ncirc,label=above:$f_5$] (g5) at (7.5,1) {};
     \node[acirc,label=below:$g_4$] (f4) at (7.5,0) {};
     \draw[->] (g5) -- (f4);

     \node[bcirc,label=above:$e_6$] (e6) at (9,2) {};
     \node[bcirc,label=below:$g_6$] (f6) at (10.5,0) {};
     \draw[->] (e6) -- (f6);

     \node[acirc,label=above:$e_7$] (e7) at (11,2) {};
     \node[ncirc,label=below:$f_6$] (g6) at (11,1) {};
     \draw[->] (e7) -- (g6);

     \node[ncirc,label=above:$f_7$] (g7) at (12,1) {};
     \node[acirc,label=below:$g_7$] (f7) at (12,0) {};
     \draw[->] (g7) -- (f7);

\draw[->] (e0) -- (e1);
\draw[->] (e1) -- (e2);
\draw[->] (e2) -- (e3);
\draw[->] (e3) -- (e4);
\draw[->] (e4) -- (e5);
\draw[->] (e5) -- (e6);
\draw[->] (e6) -- (e7);

\draw[->] (f0) -- (f1);
\draw[->] (f1) -- (f2);
\draw[->] (f2) -- (f3);
\draw[->] (f3) -- (f4);
\draw[->] (f4) -- (f5);
\draw[->] (f5) -- (f6);
\draw[->] (f6) -- (f7);

\draw[->] (g0) -- (g1);
\draw[->] (g1) -- (g2);
\draw[->] (g2) -- (g3);
\draw[->] (g3) -- (g4);
\draw[->] (g4) -- (g5);
\draw[->] (g5) -- (g6);
\draw[->] (g6) -- (g7);

      \node at (-0.5,0) {$p_3$};
      \node at (-0.5,1) {$p_2$};
      \node at (-0.5,2) {$p_1$};
                        
    \end{tikzpicture}
    \caption{A message sequence chart (MSC)\label{fig:msc}}
  \end{figure}

\subsection{MSO Logic and Its Fragments}

Next, we give an account of \emph{monadic second-order} (MSO) logic and its fragments.
Note that we restrict our attention to MSO logic interpreted \emph{over MSCs}.
We fix an infinite supply $\Var = \{x,y,\ldots\}$ of first-order variables, which range over events of an MSC, and
an infinite supply $\VAR = \{X,Y,\ldots\}$ of second-order variables, ranging over sets of events.
The syntax of $\MSO$ (we consider that $\Procs$ and $\Sigma$ are fixed) is given as follows:
\begin{align*}
\fophi & ~::=~ p(x) \mid a(x) \mid x = y \mid x \prel y \mid x \mrel y \mid x \le y \mid x \in X \mid \fophi \lor \fophi \mid \lnot \fophi \mid
\exists x. \fophi \mid \exists X. \fophi
\end{align*}
where $p \in \Procs$, $a \in \Sigma$, $x,y \in \Var$, and $X \in \VAR$.
We use the usual abbreviations to also include implication $\Longrightarrow$, conjunction $\wedge$, and universal quantification $\forall$. Moreover, the relation $x \leproc y$ can be defined by $x \le y \wedge \bigvee_{p \in \Procs} p(x) \wedge p(y)$.
We write $\Free(\Phi)$ the set of free variables of $\Phi$.

Let $M = (E,\prel,\mrel,\ploc,\lambda)$ be an MSC.
An \emph{interpretation} (for $M$) is a mapping
$\nu\colon \Var \cup \VAR \to E \cup 2^E$ assigning to each
$x \in \Var$ an event
$\nu(x) \in E$, and to each $X \in \VAR$ a set of events
$\nu(X) \subseteq E$.
We write $M,\nu \models \fophi$ if $M$ satisfies
$\fophi$ when the free variables of $\fophi$ are interpreted according to
$\nu$. Hereby, satisfaction is defined in the usual manner.
In fact, whether $M,\nu \models \fophi$ holds or not only depends on the interpretation
of variables that occur free in $\fophi$. Thus, we may restrict $\nu$ to any set of variables that contains at least all free variables.
For example, for $\fophi(x,y) = (x \mrel y)$, we have
$M,[x \mapsto e, y \mapsto f] \models \fophi(x,y)$ iff $e \mrel f$.
For a \emph{sentence} $\fophi \in \MSO$ (without free variables), we define
$\msclang(\fophi) \df \{M \in \MSCs{\Procs}{\Sigma} \mid M \models \fophi\}$.

We say that two formulas $\fophi$ and $\fophi'$ are \emph{equivalent}, written $\fophi \equiv \fophi'$, if, for all MSCs $M = (E,\prel,\mrel,\ploc,\lambda)$ and interpretations $\nu \colon \Var \cup \VAR \to E \cup 2^E$, we have $M,\nu \models \fophi$ iff $M,\nu \models \fophi'$.

Let us identify two important fragments of MSO logic: \emph{First-order} (\FO) formulas do not make use of second-order quantification (however, they may contain formulas $x \in X$). Moreover, \emph{existential} MSO ($\EMSO$) formulas are of the form $\exists X_1 \ldots \exists X_n.\fophi$ with $\fophi \in \FO$.

Let $\Class$ be $\MSO$ or $\EMSO$ or $\FO$ and let $R
\subseteq\{\prel,\mrel,\le\}$.
We obtain the logic $\Class[\Sign]$ by
restricting $\Class$ to formulas that do not make use of $\{\prel,\mrel,\le\}
\setminus \Sign$.
Note that $\Class = \Class[\prel,\mrel,\le]$.
Moreover, we let $\Lang(\Class[R]) \df \{\msclang(\fophi) \mid \fophi \in \Class[R]$ is a sentence$\}$.

Since the reflexive transitive closure of an MSO-definable binary relation is
MSO-definable, $\MSO$ and $\MSO[\prel,\mrel]$ have the same expressive power:
$\Lang(\MSO[\prel,\mrel,\leq])=\Lang(\MSO[\prel,\mrel])$.  However,
$\MSO[\le]$ (without the message relation) is strictly weaker than $\MSO$
\cite{BolligJournal}.

\begin{example}
  We give an FO formula that allows us to recover, at some event $f$, the most
  recent event $e$ that happened in the past on, say, process $p$.  More
  precisely, we define the predicate $\mathit{latest}_p(x,y)$ as $x \le y \wedge
  p(x) \wedge \forall z\bigl((z \le y \wedge p(z)) \implies z \le x\bigr)$.  The
  ``gossip language'' says that process $q$ always maintains the latest
  information that it can have about $p$.  Thus, it is defined by
  $\fophi^{\mathsf{gossip}}_{p,q} = \forall x\forall
  y.\bigl((\mathit{latest}_p(x,y) \wedge q(y)) \implies \bigvee_{a \in \Sigma}
  (a(x) \wedge a(y))\bigr) \in \FO^3[\le]$.  For example, for
  $\Procs=\{p_1,p_2,p_3\}$ and $\Sigma = \{\bbullet,\abullet,\nbullet\}$, the
  MSC $M$ from Figure~\ref{fig:msc} is contained in
  $\msclang(\fophi^{\mathsf{gossip}}_{p_1,p_3})$.  In particular, $M,[x \mapsto e_5, y
  \mapsto g_5] \models \mathit{latest}_{p_1}(x,y)$ and $\lambda(e_5) =
  \lambda(g_5) = \abullet$.
\end{example}

\subsection{Communicating Finite-State Machines}

In a communicating finite-state machine, each process $p \in \Procs$ can
perform internal actions of the form $\loc a$, where $a \in \Sigma$,
or send/receive messages from a finite set of messages $\Msg$.
A send action $\send{a}{\msg}{q}$ of process $p$ writes message
$\msg \in \Msg$ to channel $(p,q)$, and performs $a \in \Sigma$.
A receive action $\rec{a}{\msg}{q}$ reads message $\msg$ from channel $(q,p)$.
Accordingly, we let $\Act p \Msg \df \{\loc a \mid a \in \Sigma\} \cup
\{\send{a}{\msg}{q} \mid a \in \Sigma$, $m \in \Msg$, $q \in \Procs\setminus\{p\}\} \cup
\{\rec{a}{\msg}{q} \mid a \in \Sigma$, $m \in \Msg$, $q \in \Procs\setminus\{p\}\}$
denote the set of possible actions of process $p$.

A \emph{communicating finite-state machine (CFM)} over $\Procs$ and $\Sigma$
is a tuple $((\A_p)_{p \in \Procs},\Msg,\Acc)$ consisting of a finite set of
messages $\Msg$ and a finite-state transition system $\A_p=(S_p,\init_p,\Delta_p)$
for each process $p$, with finite set of states $S_p$, initial state $\init_p \in S_p$, and
transition relation $\Delta_p \subseteq S_p \times \Act p \Msg \times S_p$.
Moreover, we have an acceptance condition $\Acc \subseteq \prod_{p \in \Procs} S_p$.

Given a transition $t = (s,\act,s') \in \Delta_p$, we let $\source(t) = s$
and $\target(t) = s'$ denote the source and target states of $t$.
In addition, if $\act = \loc{a}$, then $t$ is an \emph{internal transition} and we let $\tlabel(t) = a$.
If $\act = \send{a}{\msg}{q}$, then $t$ is a \emph{send transition} and we let $\tlabel(t) = a$, $\tmsg(t) = \msg$,
and $\receiver(t) = q$.
Finally, if $\act = \rec{a}{\msg}{q}$, then $t$ is a \emph{receive transition} with $\tlabel(t) = a$, $\tmsg(t) = \msg$,
and $\sender(t) = q$.

A \emph{run} $\rho$ of $\A$ on an MSC
$M = (E,\prel,\mrel,\ploc,\lambda) \in \MSCs{\Procs}{\Sigma}$ is a mapping
associating with each event $e \in E_p$ a transition $\rho(e) \in \Delta_p$,
and satisfying the following conditions:
\begin{enumerate}
\item for all events $e \in E$, we have $\tlabel(\rho(e)) = \lambda(e)$,
\item for all $\prel$-minimal events $e \in E$, we have
  $\source(\rho(e))=\init_{p}$, where $p = \ploc(e)$,
\item for all process edges $(e,f) \in {\prel}$, we have $\target(\rho(e)) = \source(\rho(f))$,
\item for all internal events $e \in E$, $\rho(e)$ is an internal transition, and
\item for all message edges $e\mrel f$, 
  $\rho(e)$ and $\rho(f)$ are respectively send and receive transitions
  such that $\tmsg(\rho(e)) = \tmsg(\rho(f))$, $\receiver(\rho(e)) = \ploc(f)$,
  and $\sender(\rho(f)) = \ploc(e)$.
\end{enumerate}
To determine whether $\rho$ is accepting, we collect the last state $s_p$
of every process $p$.  If $E_p \neq \emptyset$, we let $s_p = \target(\rho(e))$,
where $e$ is the last event of $E_p$.  Otherwise, $s_p = \init_p$.
We say that $\rho$ is \emph{accepting} if $(s_p)_{p \in \Procs} \in \Acc$.

The \emph{language} $\msclang(\A)$ of $\A$ 
is the set of MSCs $M$ such that there
exists an accepting run of $\A$ on $M$.
Moreover, $\LangCFM \df \{\msclang(\A) \mid \A$ is a CFM$\}$.
Recall that, for these definitions, we have fixed $\Procs$ and $\Sigma$.

One of our main results states that CFMs and 
$\EMSO$ logic are expressively equivalent.  This solves a problem that was
stated as open in \cite{GKM07}:

\begin{theorem}\label{thm:main}
  $\Lang(\EMSO[\prel,\mrel,\le]) = \LangCFM$.
\end{theorem}

It is standard to prove $\LangCFM \subseteq \Lang(\EMSO[\prel,\mrel])$:
The formula guesses an assignment of transitions to events
in terms of existentially quantified second-order variables (one for each transition)
and then checks, in its first-order kernel, that the assignment is indeed an (accepting) run.
As, moreover, the class $\LangCFM$ is closed under projection, the proof of
Theorem~\ref{thm:main} comes down to the proposition below
(whose proof is spread over Sections~\ref{sec:pdl} and \ref{sec:pdl-cfm}).
Note that the translation from $\FOle$ to CFMs is inherently non-elementary,
already when $|\Procs| = 1$ \cite{phd-stockmeyer}.

\begin{proposition}\label{prop:fo-cfm}
  $\Lang(\FO[\prel,\mrel,\le]) \subseteq \LangCFM$.
\end{proposition}

\section{Star-Free Propositional Dynamic Logic}\label{sec:pdl}

In this section, we introduce a star-free version of propositional dynamic logic
and show that it is expressively equivalent to $\FO[\prel,\mrel,\leq]$.  This is
the second main result of the paper.
Then, in Section~\ref{sec:pdl-cfm}, we show how to translate star-free PDL 
formulas into CFMs.

\subsection{Syntax and Semantics}

Originally, propositional dynamic logic (PDL) has been used to reason about
program schemas and transition systems \cite{FisL79}.  Since then, PDL and its
extension with intersection and converse have developed a rich theory with
applications in artificial intelligence and verification
\cite{HalpernM92,GiacomoL94,LangeLutzJSL05,Lange06,Goeller2009}.  It has also
been applied in the context of MSCs \cite{BKM-lmcs10,Mennicke13}.

Here, we introduce a \emph{star-free} version of PDL, denoted $\sfPDLall$.
It will serve as an ``interface'' between FO logic and CFMs.
The syntax of $\sfPDLall$ and its fragment $\sfPDLm$ is given by the following grammar:
\setlength{\tabcolsep}{20pt}
\renewcommand{\arraystretch}{1.5}
\begin{center}
\fbox{
\scalebox{0.95}{
\parbox{1\textwidth}{
$\sfPDLall = \sfPDLall[\Loopname,\cup,\cap,\compl]$\\[0.8ex]
$\begin{array}{|ll|l}
\cline{1-2}
\sfPDLm & \pdlsentence ::= \E\varphi \mid \pdlsentence\vee\pdlsentence \mid \neg\pdlsentence \\
  & \varphi ::= p \mid a \mid \varphi \lor \varphi \mid \lnot \varphi
         \mid \existsp{\pi}{\varphi} \mid \Loop \pi \\[0.5ex]
  & \pi ::= {\prel} \mid {\leftmove} \mid {\mrel_{p,q}} \mid {\mrel_{p,q}^{-1}}
        \mid {\rightg \varphi} \mid {\leftg \varphi} \mid \jump p {r}
        \mid \test \varphi \mid \pi \cdot \pi & \pi \cup \pi \mid \pi \cap \pi \mid \pic \pi\\
\cline{1-2}
\end{array}$
}}}
\end{center}
where $p,r \in \Procs$, $q \in \Procs \setminus \{p\}$, and $a \in \Sigma$.  We
refer to $\pdlsentence$ as a \emph{sentence}, to $\varphi$ as an \emph{event formula}, and
to $\pi$ as a \emph{path formula}.  We name the logic star-free because we use
the operators $(\cup,\cap,\compl,\cdot)$ of star-free regular expressions instead of
the regular-expression operators $(\cup,\cdot,\ast)$ of classical PDL. However, the
formula $\rightg{\varphi}$, whose semantics is explained below, can be seen as a
restricted use of the construct $\pi^\ast$.

A sentence $\pdlsentence$ is evaluated wrt.\ an MSC $M =
(E,\prel,\mrel,\ploc,\lambda)$.
An event formula $\varphi$ is evaluated wrt.\ $M$ and an event $e \in E$.
Finally, a path formula $\pi$ is evaluated over \emph{two} events.
In other words, it
defines a binary relation $\sem M \pi \subseteq E \times E$.  We often write
$M,e,f \models \pi$ to denote $(e,f) \in \sem M \pi$.  Moreover, for $e \in E$,
we let $\sem M \pi(e) \df \{f \in E \mid (e,f) \in \sem M \pi\}$.  When $M$ is
clear from the context, we may write $\semM \pi$ instead of $\sem M \pi$.  The
semantics of sentences, event formulas, and path formulas is given in
Table~\ref{table:sem-PDL}.

\begin{table}[t]
\caption{The semantics of $\sfPDLall$\label{table:sem-PDL}}
\centering
$\begin{array}{ll}
\hline
\multicolumn{2}{l}{
M \models \E\varphi \textup{ if } M,e \models \varphi \text{ for some event } e\in E}
\\
M \models \neg\pdlsentence \textup{ if } M \not\models \pdlsentence &
M \models \pdlsentence_1 \vee \pdlsentence_2 \textup{ if } M \models \pdlsentence_1 \textup{ or } M \models \pdlsentence_2\\[0.5ex]
\hline
M,e \models p \textup{ if } \ploc(e) = p &
M,e \models \existsp{\pi}{\varphi} \textup{ if } \exists f \in \sem M \pi (e): M,f \models \varphi
\\
M,e \models a \textup{ if } \lambda(e) = a &
M,e \models \Loop \pi \textup{ if } (e,e) \in \sem M \pi
\\
M,e \models \neg\varphi \textup{ if } M,e \not\models \varphi &
M,e \models \varphi_1 \vee \varphi_2 \textup{ if } M,e \models \varphi_1 \textup{ or } M,e \models \varphi_2\\[0.5ex]
\hline
\sem M {\prel} \df \{ (e,f) \in E \times E \mid e \to f \}~~~~ &
\sem M {\mrel_{p,q}} \df \{ (e,f) \in E_p \times E_q \mid e \mrel f \}
\\
\sem M {\leftmove} \df \{ (f,e) \in E \times E \mid e \prel f \} &
\sem M {\mrel_{p,q}^{-1}} \df \{ (f,e) \in E_q \times E_p \mid e \mrel f \}
\\
\sem M {\jump p r} \df E_p \times E_r &
\sem M {\test \varphi} \df \{ (e,e) \mid e \in E : M,e \models \varphi \}
\\
\multicolumn{2}{l}{
\sem M {\rightg \varphi} \df
\{ (e,f) \in E \times E \mid e \lessproc f \textup{ and } \forall g \in E\textup{: } e \lessproc g \lessproc f \implies M,g \models \varphi\}
}
\\
\multicolumn{2}{l}{
\sem M {\leftg \varphi} \df
\{ (e,f) \in E \times E \mid f \lessproc e \textup{ and } \forall g \in E\textup{: } f \lessproc g \lessproc e \implies M,g \models \varphi\}
}
\\
\multicolumn{2}{l}{
\sem M {\pi_1 \cdot \pi_2} \df  \{ (e,g) \in E \times E \mid \exists f \in E:
                        (e,f)\in\sem{M}{\pi_1} \land (f,g)\in\sem{M}{\pi_2} \}
}
\\
\sem M {\pi_1 \cup \pi_2} \df \sem M {\pi_1} \cup \sem M {\pi_2} &
\sem M {\pic \pi} \df (E \times E) \setminus \sem M {\pi}\\
\sem M {\pi_1 \cap \pi_2} \df \sem M {\pi_1} \cap \sem M {\pi_2}\\[0.5ex]
\hline
\end{array}$
\end{table}

\begin{example}
  The usual temporal logic modalities can be expressed easily. For instance, 
  $\existsp{\prel}{\varphi}$ means that the \emph{next} event on the same 
  process satisfies $\varphi$, and $\existsp{\rightg{\varphi}}{\psi}$ corresponds 
  to the \emph{strict until} $\mathsf{X}(\varphi\mathbin{\mathsf{U}}\psi$). The 
  corresponding past modalities can be written similarly.
\end{example}

\begin{example}
Consider again the MSC $M$ from Figure~\ref{fig:msc} and the path formula
$\pi = \mrel^{-1}_{p_1,p_3} {\prel} {\mrel_{p_1,p_2}} {\prel} {\mrel_{p_2,p_3}} {\prel}$.
We have $M,g_5 \models \Loop {\pi}$. Moreover, $(e_2,e_5) \in \sem M {\rightg{{\bbullet}}}$ but
$(e_2,e_6) \not\in \sem M {\rightg{{\bbullet}}}$.
\end{example}

We use the usual abbreviations for sentences and event formulas such as
implication and conjunction.  Moreover, $\True \df p \vee \neg p$ (for some
arbitrary process $p \in \Procs$) and $\False \df \neg\True$.  Finally, we
define the event formula $\existsptrue{\pi} \df \existsp{\pi}{\True}$, and the
path formulas ${\rightp} \df {\rightg \True}$ and ${\righta} \df {\rightp} \cup
\test \True$.

Note that there are some redundancies in the logic. For example (letting $\equiv$ denote logical equivalence), ${\prel} \equiv {\rightg{\False}}$, $\pi_1 \cap \pi_2 \equiv \pic{(\pic \pi_1 \cup \pic \pi_2)}$, and $\Loop \pi \equiv \existsptrue{\test{\True} \cap \pi}$.  Some of them are necessary to define certain subclasses of $\sfPDLall$.  For every $R \subseteq
\{\Loopname,\cup,\cap,\compl\}$, we let $\sfPDLall[R]$ denote the fragment of
$\sfPDLall$ that does not make use of $\{\Loopname,\cup,\cap,\compl\} \setminus
R$.  In particular, $\sfPDLall = \sfPDLall[\Loopname,\cup,\cap,\compl]$.  Note
that, syntactically, ${\righta}$ is not contained in $\sfPDLm$ since union is 
not permitted.

Note that $\sfPDLall[\cup]$ over MSCs is analogous to Conditional XPath over trees \cite{Marx05}.\footnote{Thanks to Sylvain Schmitz for pointing this out.} However, while Marx showed that Conditional XPath is expressively complete for FO logic over
ordered unranked trees, our expressive completeness result over MSCs crucially relies on the $\Loopname$ modality.

\subsection{Main Results}

Let $\FO^3[\prel,\mrel,\le]$ be the set of formulas from $\FO[\prel,\mrel,\le]$
that use at most three different first-order
variables (however, a variable can be quantified and reused several times in a
formula). The main result of this section is that, for formulas with zero or one
free variable, the logics
$\FO[\prel,\mrel,\le]$, $\FO^3[\prel,\mrel,\le]$, $\PDL$, and $\PDLm$ are
expressively equivalent.

Consider $\FO[\prel,\mrel,\le]$ formulas $\Phi_0$, $\Phi_1(x)$ and $\Phi_2(x,y)$
with respectively zero, one, and two free variables (hence, $\Phi_0$ is a
sentence).  Consider also some $\PDL$ sentence $\pdlsentence$, event formula
$\varphi$, and path formula $\pi$.  The respective formulas are equivalent,
written $\Phi_0\equiv\pdlsentence$, $\Phi_1(x)\equiv\varphi$, and
$\Phi_2(x,y)\equiv\pi$, if, for all MSCs $M$ and all events $e,f$ in
$M$, we have
\begin{align*}
  M &\models \Phi_0 && \text{iff} & M &\models \pdlsentence \\
  M,[x\mapsto e] &\models \Phi_1(x) && \text{iff} & M,e &\models \varphi \\
  M,[x \mapsto e, y \mapsto f] &\models \Phi_2(x,y) && \text{iff} & M,e,f &\models \pi
\end{align*}

We start with a simple observation, which can be shown easily by induction:

\begin{proposition}\label{prop:pdl-fo3}
  Every $\PDL$ formula is equivalent to some $\FO^3[\prel,\mrel,\le]$ formula.
  More precisely, for every $\PDL$ sentence $\pdlsentence$, event formula $\varphi$,
  and path formula $\pi$, there exist some $\FO^3[\prel,\mrel,\le]$ sentence
  $\widetilde{\pdlsentence}$, formula $\phifo\varphi(x)$ with one free variable, and
  formula $\pifo\pi(x,y)$ with two free variables, respectively, such that, 
  $\pdlsentence\equiv\widetilde{\pdlsentence}$, $\varphi\equiv\phifo\varphi(x)$,
  and $\pi\equiv\pifo\pi(x,y)$.
\end{proposition}

The main result is a \emph{strong} converse of Proposition~\ref{prop:pdl-fo3}:

\begin{theorem}\label{thm:FO-to-PDLm-main}
  Every $\FO[\prel,\mrel,\le]$ formula with at most two free variables is
  equivalent to some $\PDL$ formula.
  More precisely, for every
  $\FO[\prel,\mrel,\le]$ sentence $\Phi_0$, formula $\Phi_1(x)$ with one free
  variable, and formula $\Phi_2(x,y)$ with two free variables, there exist some
  $\PDLm$ sentence $\pdlsentence$, $\PDLm$ event formula $\varphi$, and $\PDLm$
  path formulas $\pi_{ij}$, respectively, such that, $\Phi_0\equiv\pdlsentence$,
  $\Phi_1(x)\equiv\varphi$, and $\Phi_2(x,y)\equiv\bigcup_i\bigcap_j\pi_{ij}$.
\end{theorem}

From Theorem~\ref{thm:FO-to-PDLm-main} and Proposition~\ref{prop:pdl-fo3},
we deduce that $\FO$ has the three variable property:

\begin{corollary}\label{thm:fo3}
  $\Lang(\FO[\prel,\mrel,\le]) = \Lang(\FO^3[\prel,\mrel,\le])$.
\end{corollary}

% ------------------------------------------------------------

\subsection{From $\FO$ to $\boldsymbol{\PDL}$}\label{sec:fo-pdl}

In the remainder of this section, we give the translation from $\FO$ to $\PDL$.
We start with some basic properties of $\PDL$. First, the converse of a $\PDL$ formula is definable
in $\PDL$ (easy induction on $\pi$).

\begin{lemma}\label{lem:inverse}
  Let $R \subseteq \{\Loopname,\cup,\cap,\compl\}$ and $\pi \in \sfPDL[R]$ be a path formula.
  There exists $\pi^{-1} \in \sfPDL[R]$ such that, for all MSCs $M$,
  $\sem M {\pi^{-1}} = \sem M {\pi}^{-1} = \{(f,e) \mid (e,f) \in \sem M {\pi}\}$.
\end{lemma}

Given a $\PDLm$ path formula $\pi$, we denote by $\Comp{\pi}$ the set of 
pairs $(p,q) \in \Procs \times \Procs$ such that there may be a $\pi$-path from some event 
on process $p$ to some event on process $q$. Formally, we let
$\Comp{\prel}=\Comp{\leftmove}=\Comp{\rightg{\varphi}}=\Comp{\leftg{\varphi}}=\Comp{\test{\varphi}}=\id$,
where $\id=\{(p,p)\mid p\in\Procs\}$; 
$\Comp{\mrel_{p,q}}=\Comp{\mrel_{q,p}^{-1}}=\{(p,q)\}$; $\Comp{\jump{p}{r}}=\{(p,r)\}$; and
$\Comp{\pi_1\cdot\pi_2}=\Comp{\pi_2}\circ\Comp{\pi_1}=\{(p,r)\mid \exists q:
(p,q)\in\Comp{\pi_1}, (q,r)\in\Comp{\pi_2}\}$.

Notice that, for all path formulas $\pi \in \PDLm$, the relation $\Comp{\pi}$ is
either empty or a singleton $\{(p,q)\}$ or the identity $\id$.  Moreover,
$M,e,f\models\pi$ implies $(\ploc(e),\ploc(f))\in\Comp{\pi}$. 
Therefore, all events in $\semM{\pi}(e)$ are on the same process, and if this 
set is nonempty (i.e., if $M,e\models\existsptrue{\pi}$), then 
$\minpie{\pi}{e}$ and $\maxpie{\pi}{e}$ are well-defined.

\begin{example}
Consider the MSC from Figure~\ref{fig:msc} and
$\pi = {\rightp}{\mrel_{p_1,p_2}}{\prel}{\mrel_{p_2,p_3}}{\prel}$.
We have $\Comp{\pi} = \{(p_1,p_3)\}$. Moreover, $\minpie{\pi}{e_2} = g_4$ and $\maxpie{\pi}{e_2} = g_5$.
\end{example}

We say that $\pi \in \PDLm$ is \emph{monotone} if, for all MSCs $M$ and events 
$e,f$ such that $M,e \models \existsptrue \pi$, $M,f \models \existsptrue \pi$,
and $e \leproc f$, we have $\minpie \pi e \leproc \minpie \pi f$ and $\maxpie
\pi e \leproc \maxpie \pi f$. 
Lemmas~\ref{lem:min-conc} and \ref{lem:monotone} are shown by simultaneous induction.

\begin{lemma}\label{lem:min-conc}
  Let $\pi_1,\pi_2 \in \PDLm$ be path formulas. 
  For all MSCs $M$ and events $e$ such that $M,e\models\existsptrue{\pi_1 \cdot \pi_2}$, we have
  \begin{align*}
    \minpie{\pi_1 \cdot \pi_2}{e} & = \minpie {\pi_2} {\minpie {\pi_1 \cdot \test{\existsptrue {\pi_2}}} {e}} \text{ and }
      \\
    \maxpie{\pi_1 \cdot \pi_2}{e} & = \maxpie {\pi_2} {\maxpie {\pi_1 \cdot \test{\existsptrue {\pi_2}}} {e}} \, .
  \end{align*}
\end{lemma}

\begin{lemma}\label{lem:monotone}
  All $\PDLm$ path formulas are monotone.
\end{lemma}

\begin{proof}[Proof of Lemma~\ref{lem:min-conc} and Lemma~\ref{lem:monotone}]
  We first show that Lemma~\ref{lem:min-conc} holds when $\pi_2$ is monotone.
  We then use this to prove by induction that all $\sfPDL$ formulas are monotone
  (Lemma~\ref{lem:monotone}). Therefore, we deduce that Lemma~\ref{lem:min-conc} is always
  true.

  Let $\pi_1,\pi_2 \in \PDLm$ be path formulas such that $\pi_2$ is monotone. 
  Let $M$ be an MSCs and $e$ be some event in $M$ such that
  $M,e\models\existsptrue{\pi_1\pi_2}$.
  The proof is illustrated in Figure~\ref{fig:min-conc}.
  We let $g = \minpie {\pi_1 \pi_2} {e}$.  Since $M,e,g \models \pi_1\pi_2$,
  there exists $f$ such that $M,e,f \models \pi_1$, and $M,f,g \models \pi_2$.
  In particular, $M,e,f \models \pi_1 \test{\existsptrue {\pi_2}}$, so $f' =
  \minpie {\pi_1 \test{\existsptrue {\pi_2}}} {e}$ is well-defined and $f'
  \leproc f$.  Since $\pi_2$ is monotone, $\minpie {\pi_2} {f'} \leproc
  \minpie {\pi_2} {f} \leproc g$.
  Also, $M,e,\minpie {\pi_2} {f'} \models \pi_1\pi_2$.  Hence $g \leproc \minpie {\pi_2} {f'}$.
  Therefore, $g = \minpie {\pi_2} {f'}$.
  
  The proof that $\maxpie {\pi_1 \pi_2} {e} =
  \maxpie {\pi_2} {\maxpie {\pi_1 \test{\existsptrue {\pi_2}}} {e}}$
  is similar.
  
  \begin{figure}[h]
  \centering
  \begin{tikzpicture}[scale=1,font=\footnotesize,inner sep=1pt,auto,>=stealth]
      \node[dot,label=below:{$e$}] (e) at (0,0) {};
      \node[dot,label=right:{$f$}] (f) at (1,1) {};
      \node[dot,label=left:{$\minpie {\pi_1 \test{\existsptrue {\pi_2}}} {e} =: f'$}] (fp) at (-1,1) {};
      \node[dot,label=left:{$\minpie {\pi_2} {f'}$}] (gp) at (-1.1,2) {};
      \node[dot,label=right:{\;\,$g := \minpie {\pi_1 \pi_2} {e}$}] (g) at (0,2) {};
      \node (leproc) at (0,1) {$\leproc$};
      \node at (-0.5,1.95) {$\leproc$};
      
      \path
      (e) edge[decorate,decoration={snake,post length=0.5mm,amplitude=0.4mm},->=2] node[right] {\;$\pi_1$} (f)
      (e) edge[decorate,decoration={snake,post length=0.5mm,amplitude=0.4mm},->=2] node[left] {$\pi_1$\;} (fp)
      (f) edge[decorate,decoration={snake,post length=0.5mm,amplitude=0.4mm},->=2] node[right] {\;$\pi_2$} (g)
      (fp) edge[decorate,decoration={snake,post length=0.5mm,amplitude=0.4mm},->=2] node[left] {$\pi_2$\,} (gp);       
    \end{tikzpicture}
    \caption{Proof of Lemma~\ref{lem:min-conc}\label{fig:min-conc}}
   \end{figure}

  We turn now to the proof of Lemma~\ref{lem:monotone}.  Actually, we prove a
  slightly stronger statement.  We show by induction on $\pi$ that, for all
  $\PDLm$ event formulas $\psi$, the path formula $\pi \cdot \test{\psi}$ is
  monotone.

  Let $e,f$ be events such that $e\leproc f$, $M,e\models\pi\cdot\test{\psi}$
  and $M,f\models\pi\cdot\test{\psi}$.  Let $e' = \minpie{\pi\cdot\test{\psi}}
  e$ and $f' = \minpie{\pi\cdot\test{\psi}} f$.  We show that $e'\leproc f'$.
  The proof that
  $\maxpie{\pi\cdot\test{\psi}}{e}\leproc\maxpie{\pi\cdot\test{\psi}}{e}$ is
  similar.  We start with the base cases.
  
  If $\pi = \test {\varphi}$, we have $e' = e \leproc f = f'$.
  The result is also trivial for $\pi={\prel}$ or $\pi={\leftmove}$.
  It follows from the fact that channels are FIFO for $\pi={\mrel_{p,q}}$ or 
  $\pi={\mrel_{p,q}^{-1}}$. When $\pi=\jump{p}{q}$ we have $e'=f'$. Suppose 
  that $\pi={\rightg{\varphi}}$. It is easy to see that either $e'\leproc 
  f\ltproc f'$ or $e'=f'$.  Similarly, when $\pi={\leftg{\varphi}}$ we have
  either $e'\ltproc e\leproc f'$ or $e'=f'$.
  
  The proof for $\pi=\pi_1\cdot\pi_2$ is illustrated in Figure~\ref{fig:monotone}.

  \begin{figure}[h]
  \centering
  \begin{tikzpicture}[scale=1,font=\footnotesize,inner sep=1pt,auto,>=stealth]
       \node[dot,label=below:{$e$}] (e) at (0,0) {};
       \node[dot,label=below:{$f$}] (f) at (2,0) {};
       \node (leproc) at (1,0) {$\leproc$};

       \node[dot,label=left:{$\minpie{\pi_1\cdot\test{\existsp{\pi_2}{\psi}}}{e} =: e''$}] (epp) at (-0.5,1) {};
       \node[dot,label=right:{$f'':=\minpie{\pi_1\cdot\test{\existsp{\pi_2}{\psi}}}{f}$}] (fpp) at (1.5,1) {};
       \node (leproc) at (0.5,1) {$\leproc$};       

       \node[dot,label=left:{$\minpie{\pi_2\cdot\test{\psi}}{e''} = e'$}] (ep) at (-0.5,2) {};
       \node[dot,label=right:{$f'=\minpie{\pi_2\cdot\test{\psi}}{f''}$}] (fp) at (1.5,2) {};
       \node (leproc) at (0.5,2) {$\leproc$};       
       
       \path
       (e) edge[decorate,decoration={snake,post length=0.5mm,amplitude=0.4mm},->=2] node[left] {$\pi_1$\,} (epp)
       (f) edge[decorate,decoration={snake,post length=0.5mm,amplitude=0.4mm},->=2] node[right] {\,$\pi_1$} (fpp)
       (epp) edge[decorate,decoration={snake,post length=0.5mm,amplitude=0.4mm},->=2] node[left] {$\pi_2$\,} (ep)
       (fpp) edge[decorate,decoration={snake,post length=0.5mm,amplitude=0.4mm},->=2] node[right] {\,$\pi_2$} (fp);
     \end{tikzpicture}
    \caption{Proof of Lemma~\ref{lem:monotone}\label{fig:monotone}}
   \end{figure}

  By induction, the path formula $\pi_2\cdot\test{\psi}$ is monotone.  So we can
  apply the special case of Lemma~\ref{lem:min-conc} proved above to the product
  $\pi_1\cdot(\pi_2\cdot\test{\psi})$.  Let
  $e''=\minpie{\pi_1\cdot\test{\existsp{\pi_2}{\psi}}}{e}$ and
  $f''=\minpie{\pi_1\cdot\test{\existsp{\pi_2}{\psi}}}{f}$.  We have
  $e'=\minpie{\pi_2\cdot\test{\psi}}{e''}$ and
  $f'=\minpie{\pi_2\cdot\test{\psi}}{f''}$.
  Again by induction, the path formula  $\pi_1 \cdot \test{\existsp{\pi_2}{\psi}}$ 
  is monotone and we obtain $e''\leproc f''$.
  We get $e'\leproc f'$ since $\pi_2\cdot\test{\psi}$ is monotone.
\end{proof}

The following crucial lemma states that, for all path formulas $\pi \in \PDLm$
and events~$e$ in some MSC, $\semM \pi (e)$ contains precisely the events
that lie in the interval between $\minpie \pi e$ and $\maxpie \pi e$ and
that satisfy $\existsptrue{\pi^{-1}}$.

\begin{lemma}\label{lem:image}
  Let $\pi$ be a $\PDLm$ path formula. For all MSCs $M$ and events $e$ such
  that $M,e\models\existsptrue\pi$, we have
  \[
    \semM \pi (e) = \{ f \in E \mid \minpie \pi e \leproc f \leproc
    \maxpie \pi e \land M,f \models \existsptrue {\pi^{-1}} \} \, .
  \]
\end{lemma}

\begin{proof}
  The left-to-right inclusion is trivial.
  We prove the right-to-left inclusion by induction on $\pi$.
  The base cases are immediate.
  
  Assume that $\pi = \pi_1 \cdot \pi_2$. For illustration, consider Figure~\ref{fig:image}.
  We let $f_1 = \minpie {\pi_1 \test{\existsptrue{\pi_2}}} {e}$, 
  $f_2 = \minpie {\pi_2} {f_1}$,
  $g_1 = \maxpie {\pi_1 \test{\existsptrue{\pi_2}}} {e}$, 
  and $g_2 = \maxpie {\pi_2} {g_1}$.  
  By Lemma~\ref{lem:min-conc},
  we have $f_2=\minpie{\pi_1\pi_2}{e}$ and
  $g_2=\maxpie{\pi_1\pi_2}{e}$.
  Let $h_2 \in E$ such that $f_2 \leproc h_2 \leproc g_2$ and
  $M,h_2 \models \existsptrue{(\pi_1\pi_2)^{-1}}$.
  If $h_2 \leproc \maxpie {\pi_2} {f_1}$, then by induction hypothesis,
  $M,f_1,h_2 \models \pi_2$, and we obtain $M,e,h_2 \models \pi_1\pi_2$.
  Similarly, if $\minpie {\pi_2} {g_1} \leproc h_2$, then $M,g_1,h_2 \models \pi_2$
  and $M,e,h_2 \models \pi_1\pi_2$.
  So assume $\maxpie {\pi_2} {f_1} \ltproc h_2 \ltproc \minpie {\pi_2} {g_1}$.
  Since $M,h_2 \models \existsptrue{\pi_2^{-1}\pi_1^{-1}}$, there exists
  $h_1$ such that $M,h_1,h_2 \models \pi_2$ and
  $M,h_1 \models \existsptrue{\pi_1^{-1}}$.
  Moreover, $\minpie {\pi_2} {h_1} \leproc h_2 \ltproc \minpie {\pi_2} {g_1}$, hence
  $h_1 \leproc g_1$ by Lemma~\ref{lem:monotone} (notice that $g_1$ and 
  $h_1$ must be on the same process).
  Similarly, $\maxpie {\pi_2} {f_1} \ltproc h_2 \leproc \maxpie {\pi_2} {h_1}$,
  hence $f_1 \leproc h_1$. We then have $f_1 \leproc h_1 \leproc g_1$, and
  $M,h_1 \models \existsptrue{\pi_1^{-1}}$. By induction hypothesis,
  $M,e,h_1 \models \pi_1$. Hence, $M,e,h_2 \models \pi_1\pi_2$.
\end{proof}

  \begin{figure}[h]
  \centering
  \begin{tikzpicture}[scale=0.75,baseline,font=\footnotesize,inner sep=3pt,auto,>=stealth]
     
     \node[dot,label=below:{$e$}] (e) at (0,0) {};
     \node[dot,label=left:{$h_1$}] (h1) at (0,2) {};
     \node[dot,label=above:{$h_2$}] (h2) at (0,4) {};
     \node[dot,label=left:{$\minpie {\pi_1 \test{\existsptrue{\pi_2}}} {e} =: f_1$}] (f1) at (-3,2) {};
     \node[dot,label=right:{$g_1 := \maxpie {\pi_1 \test{\existsptrue{\pi_2}}} {e}$}] (g1) at (3,2) {};

     \node[dot,label={[anchor=east,xshift=1.5ex]above:{$\begin{aligned}\minpie {\pi_2} {f_1} & =: f_2 \\ \minpie{\pi_1\pi_2}{e} & = \end{aligned}$}}] (f2) at (-5.5,4) {};
     \node[dot,label={[anchor=west,xshift=-1.5ex]above:{$\begin{aligned}g_2 & := \maxpie {\pi_2} {g_1} \\ & = \maxpie{\pi_1\pi_2}{e} \end{aligned}$}}] (g2) at (5.5,4) {};
     
     \node[dot,label=above:{$\maxpie {\pi_2} {f_1}$}] (max) at (-2.5,4) {};
     \node[dot,label=above:{$\minpie {\pi_2} {g_1}$}] (min) at (2.5,4) {};
            
       \node (leproc) at (-1.5,2) {$\leproc$};
       \node (leproc) at (1.5,2) {$\leproc$};
       \node (leproc) at (-1,4) {$\ltproc$};
       \node (leproc) at (1,4) {$\ltproc$};
       \node (leproc) at (-4,4) {$\leproc$};
       \node (leproc) at (4,4) {$\leproc$};

       \path[inner sep=1pt]
       (e) edge[decorate,decoration={snake,post length=0.5mm,amplitude=0.4mm},->=2,bend left=0] node[left] {$\pi_1$~~~} (f1)
       (e) edge[decorate,decoration={snake,post length=0.5mm,amplitude=0.4mm},->=2,bend left=-0] node[right] {~~~$\pi_1$} (g1)
       (h1) edge[decorate,decoration={snake,post length=0.5mm,amplitude=0.4mm},->=2,bend left=0] node[right] {~$\pi_1^{-1}$} (e)
       (h2) edge[decorate,decoration={snake,post length=0.5mm,amplitude=0.4mm},->=2,bend left=0] node[right] {~$\pi_2^{-1}$} (h1)
       (f1) edge[decorate,decoration={snake,post length=0.5mm,amplitude=0.4mm},->=2,bend left=0] node[right] {~$\pi_2$} (max)
       (g1) edge[decorate,decoration={snake,post length=0.5mm,amplitude=0.4mm},->=2,bend left=0] node[left] {$\pi_2$~~} (min)
       (f1) edge[decorate,decoration={snake,post length=0.5mm,amplitude=0.4mm},->=2,bend left=0] node[left] {$\pi_2$~~} (f2)
       (g1) edge[decorate,decoration={snake,post length=0.5mm,amplitude=0.4mm},->=2,bend left=0] node[right] {~~$\pi_2$} (g2);

     \end{tikzpicture}
     \caption{Proof of Lemma~\ref{lem:image}\label{fig:image}}
   \end{figure}

Using Lemma~\ref{lem:image}, we can give a characterization of
$\semM {\pic \pi} (e)$ (when $\pi \in \PDLm$) that also relies on intervals
delimited by $\minpie \pi e$ and $\maxpie \pi e$. More precisely,
$\semM {\pic \pi} (e)$ is the union of the following sets (see Figure~\ref{fig:charcompl}):
(i) the interval of all events to the left of $\minpie \pi e$,
(ii) the interval of all events to the right of $\maxpie \pi e$,
(iii) the set of events located between $\minpie \pi e$ and $\maxpie \pi e$ and
satisfying $\lnot \existsptrue{\pi^{-1}}$,
(iv) all events located on other processes than $\minpie \pi e$ and
$\maxpie \pi e$.

\begin{figure}[h]
\centering
   \begin{tikzpicture}[auto,>=stealth,inner sep=2pt,scale=0.6]
    \node[dot,label=below:{$e$}] (e) at (3,-2) {};
    \begin{scope}[mygreen]
      \foreach \i in {0,1,3,6} {
        \node[dot,label={above:{\footnotesize $\existsptrue{\pi^{-1}}$}}]
        (e\i) at (1.3*\i,0) {};
      }
      \path[->]
      (e) edge[decorate,decoration={snake,post length=0.7mm,amplitude=0.4mm}]
      node[below left] {$\minpi \pi$} (e0)
      (e) edge[decorate,decoration={snake,post length=0.7mm,amplitude=0.4mm}]
      node[below right] {$\maxpi \pi$} (e6) ;
    \end{scope}
    \foreach \i in {-1,2,4,5,7,8} {
      \node[dot] (e\i) at (1.3*\i,0) {};
    }
    \foreach \i [evaluate=\i as \j using int(\i+1)] in {-1,...,7} {
      \draw[->] (e\i) -- (e\j) ;
    }

    \begin{pgfonlayer}{background}
      \node[rectangle,minimum size=0.5cm,rounded corners=0.1cm,
      fill=purple!40,label={[purple]above left:{(i)}}] at (e-1) {};

      \node[minimum size=0.5cm,rounded corners=0.1cm, fill=colj!40]
      at (e2) {};
      \fill[colj!40,rounded corners=1mm]
      ($(e4)+(-0.4,0.4)$) rectangle ($(e5)+(0.4,-0.4)$) ;
      \node[colj,below=0.1cm of e3] {(iii)} ;

      \fill[coll!40,rounded corners=1mm]
      ($(e7)+(-0.4,0.4)$) rectangle ($(e8)+(0.4,-0.4)$) ;
      \node[coll,above right=0.2cm of e8] {(ii)};
    \end{pgfonlayer}
  \end{tikzpicture}
  \caption{Characterization of $\semM {\pic \pi} (e)$ for $\pi \in \PDLm$\label{fig:charcompl}}
\end{figure}

This description of $\semM {\pic \pi} (e)$ can be used to rewrite
$\pic \pi$ as a union of $\PDLm$ formulas.
In a first step, we show that, if $\pi$ is a $\PDLm$
formula, then the relation $\{(e,\minpie \pi e)\}$ can also be expressed
in $\PDLm$ (and similarly for $\max$).

\begin{lemma}
Let $R = \emptyset$ or $R = \{\Loopname\}$.
  For every path formula $\pi \in \sfPDLall[R]$, there exist $\sfPDLall[R]$ path formulas
  $\minpi \pi$ and $\maxpi \pi$ such that $M,e,f \models \minpi \pi$
  iff $f = \minpie \pi e$, and $M,e,f \models \maxpi \pi$
  iff $f = \maxpie \pi e$.
\end{lemma}

\begin{proof}
  We construct, by induction on $\pi$, formulas $\minpi {(\pi \cdot\test{\psi})}$
  for all $\sfPDLall[R]$ event formulas $\psi$.
  For $\pi \in \{ {\prel}, {\leftmove}, {\mrel_{p,q}}, {\mrel^{-1}_{p,q}},
  \test {\varphi}\}$, we let $\minpi {(\pi \cdot\test{\psi})} = \pi \cdot\test{\psi}$.
  Then,
  \begin{align*}
    \minpi {({\rightg \varphi} \cdot\test{\psi})}
    & = {\rightg {\varphi \land \lnot \psi}} \cdot\test{\psi} \\
    \minpi {({\leftg \varphi} \cdot\test{\psi})}
    & = {\leftg \varphi}\cdot \test {\psi \land (\lnot \varphi \lor \lnot
      \existsp{\leftg \varphi}{\psi}) } \\
    \minpi {(\jump p q \cdot\test{\psi})}
    & = \jump p q \cdot\test{\psi \land \lnot \existsp{\leftp} \psi} \\
    \minpi {(\pi_1\cdot\pi_2 \cdot\test{\psi})}
    & = \minpi {(\pi_1 \cdot\test{\existsp{\pi_2}{\psi}})}
      \cdot \minpi {(\pi_2 \cdot\test{\psi})} \, .
  \end{align*}
  The construction of $\maxpi \pi$ is similar.
\end{proof}

We are now ready to prove that any boolean combination of $\PDLm$ formulas
is equivalent to a positive one, i.e., one that does not use complement.

\begin{lemma}\label{lem:closure-complement}
  For all path formulas $\pi \in \PDLm$, there exist $\PDLm$ path formulas
  $(\pi_i)_{1 \le i \le |\Procs|^2+3}$ such that
  $\pic \pi \equiv \bigcup_{1 \le i \le |\Procs|^2+3} \pi_i$.
\end{lemma}

\begin{proof}
  We show $\pic \pi \equiv \sigma$, where
  \[
  \sigma = (\minpi \pi \cdot {\leftp}) \cup (\maxpi \pi \cdot {\rightp}) 
  \cup ( \pi \cdot {\rightp} \cdot \test{\neg\existsptrue{\pi^{-1}}}) 
  \cup \bigcup_{(p,q) \in \Procs^2} \test{\lnot\existsp{\pi}{q}} \cdot \jump p q \,.
  \]
  
  Let $M = (E,\prel,\mrel,\ploc,\lambda)$ be an MSC and $e,f \in E$.
  We write $p = \ploc(e)$, $q = \ploc(f)$.
  Let us show that $M,e,f \models \pi^c$ iff $M,e,f \models \sigma$.
  If $M,e \models \lnot\existsp{\pi}{q}$, then both
  $M,e,f \models \pic \pi$ and $M,e,f \models \sigma$ hold.
  In the following, we assume that $M,e \models \existsp{\pi}{q}$,
  and thus that $\minpie \pi e$ and $\maxpie \pi e$ are well-defined
  and on process~$q$.
  Again, if $f \ltproc \minpie \pi e$ or $\maxpie \pi e \ltproc f$, then
  both $M,e,f \models \pic \pi$ and $M,e,f \models \sigma$ hold.
  And if $\minpie \pi e \leproc f \leproc \maxpie \pi e$, then, by
  Lemma~\ref{lem:image}, we have
  $M,e,f \models \pi^c$ iff $M,f \models \lnot\existsptrue{\pi^{-1}}$,
  iff $M,e,f \models \sigma$.
\end{proof}

\newcommand{\bunderbrace}[2]{%
  \begin{array}[t]{@{}c@{}}
  \underbrace{#1}\\
  #2
  \end{array}
}

The rest of this section is dedicated to the proof of Theorem~\ref{thm:FO-to-PDLm-main},
stating that every $\FO[\prel,\mrel,\le]$ formula with at most two free variables can
be translated into an equivalent $\PDL$ formula.
As we proceed by induction, we actually need a more general statement, which takes into account arbitrarily many free variables:

\begin{proposition}\label{thm:FO-to-PDLm}
  Every formula $\Phi\in \FOle$ with at least one free variable
  is equivalent to a 
  boolean combination of formulas of the form
  $\pifo \pi(x,y)$, where $\pi \in \sfPDLm$ and $x,y \in \Free(\Phi)$.
\end{proposition}

\begin{proof} 
  In the following, we will simply write $\pi(x,y)$ for $\pifo\pi(x,y)$, where
  $\pifo\pi(x,y)$ is the $\FO$ formula equivalent to $\pi$ as defined in
  Proposition~\ref{prop:pdl-fo3}.
  The proof is by induction. For convenience, we assume that $\Phi$ is in 
  prenex normal form.
  If $\Phi$ is quantifier free, then it is a boolean
  combination of atomic formulas.  For $x,y \in \Var$, atomic formulas are
  translated as follows:
    \[
      \begin{array}{rclcrclcrcl}
        p(x) & \equiv & \test{p} (x,x)
        & \qquad\qquad & x \prel y & \equiv & {\prel}(x,y)
        & & x = y & \equiv & \test{\True}(x,y)\\[1.5ex]
        a(x) & \equiv & \test{a} (x,x)
        & & x \mrel y & \equiv &
        \displaystyle \bigvee_{(p,q) \in \Ch} {\mrel_{p,q}}(x,y)
      \end{array}
    \]
Moreover, $x \le y$ is equivalent to the disjunction of the formulas
$\bigl(\pi \cdot {\mrel_{p_1,p_2}} \cdot {\rightp} \cdot {\mrel_{p_2,p_3}}
        \cdots {\rightp} \cdot {\mrel_{p_{m-1},p_{m}}} \cdot \pi'\bigr)(x,y)$,
where $1 \le m \le |\Procs|$, $p_1,\ldots,p_m \in \Procs$ are such that $p_i 
\neq p_{j}$ for all $1\leq i<j\leq m$, and $\pi,\pi' \in \{\rightp,\test{\True}\}$.%        

  \paragraph{Universal quantification.} We have $\forall x.\Psi \equiv 
  \neg\exists x.\neg\Psi$. Since we allow boolean combinations, dealing with 
  negation is trivial. Hence, this case reduces to existential quantification.

  \paragraph{Existential quantification.} Suppose that $\Phi = \exists x.\Psi$.
  If $x$ is not free in $\Psi$, then $\Phi\equiv\Psi$ and we are done by 
  induction. Otherwise, assume that $\Free(\Psi)=\{x_1,\ldots,x_n\}$ with $n>1$ 
  and that $x=x_n$. 
  By induction, $\Psi$ is equivalent to a boolean combination 
  of formulas of the form $\pi(y,z)$ with $y,z\in\Free(\Psi)$.
  We transform it into a finite disjunction of formulas of the form
  $\bigwedge_{j} \pi_{j}(y_{j},z_{j})$, where $y_j = x_{i_1}$ and $z_j =
  x_{i_2}$ for some $i_1 \le i_2$.
  To do so, we first eliminate negation using Lemma~\ref{lem:closure-complement}.
The resulting positive boolean combination is then brought into disjunctive normal form.
Note that this latter step may cause an exponential blow-up so that the overall construction
is nonelementary (which is unavoidable \cite{phd-stockmeyer}).
Finally, the variable ordering can be guaranteed by replacing $\pi_j$ with $\pi_j^{-1}$ whenever needed.
    
  Now, $\Phi=\exists x_n.\Psi$ is equivalent to a finite disjunction of formulas
  of the form
  $$\bigwedge_{j \in I} \pi_{j} (y_{j},z_{j})
      ~\land~ \bunderbrace{\exists x_n. \Bigl(\bigwedge_{j \in J} \pi_{j} (y_{j},x_n) \land
      \bigwedge_{j \in J'} \pi_{j} (x_n,x_n)\Bigr)}{=: \Upsilon}$$
  for three finite, pairwise disjoint index sets $I,J,J'$ such that $y_j \in
  \{x_1,\ldots,x_{n-1}\}$ for all $j \in I \cup J$, and $z_j \in
  \{x_1,\ldots,x_{n-1}\}$ for all $j \in I$.  
  Notice that $\Free(\Upsilon)\subseteq\{x_1,\ldots,x_{n-1}\}$. 
  If $J = \emptyset$, then\footnote{In this case, $\Upsilon$ is a sentence
  whereas $x_1$ is free in the right hand side.  Notice that $\equiv$ does not
  require the two formulas to have the same free variables.}
  $$
  \Upsilon \equiv \bigvee_{p,q \in P}
  \Big( \jump p q \cdot \test{\bigwedge_{j \in J'} \Loop{\pi_{j}}}
  \cdot \jump q p\Big) (x_1,x_1) \, .
  $$
  So assume $J \neq \emptyset$. Set
  $$\Upsilon' \df
      \bigvee_{k,\ell \in J}
      \left(\begin{array}{rl}
      & \bigwedge_{j \in J} ((\minpi {\pi_{j}}) \cdot {\righta} \cdot
      (\minpi {\pi_{k}})^{-1}) (y_{j},y_{k})
      \\ \land & \bigwedge_{j \in J} ((\maxpi {\pi_{\ell}}) \cdot {\righta}
      \cdot (\maxpi {\pi_{j}})^{-1}) (y_{\ell},y_{j})\\
      \land &
      (\pi_{k} \cdot \test{\psi} \cdot \pi_{\ell}^{-1}) (y_{k},y_{\ell})
      \end{array}\right)$$
  where $ \psi = \bigwedge_{j \in J} \existsptrue{\pi_{j}^{-1}}
  \land \bigwedge_{j \in J'} \Loop {\pi_{j}}$.
  We have $\Free(\Upsilon')=\Free(\Upsilon)\subseteq\{x_1,\ldots,x_{n-1}\}$.
    
  \begin{claim}\label{claim:fo-pdl}
    We have $\Upsilon \equiv \Upsilon'$.
  \end{claim}

  Intuitively, by Lemma~\ref{lem:image}, we know that
  $\Upsilon$ holds iff the intersection of the intervals
  $[\minpie {\pi_j} {y_j}, \maxpie {\pi_j} {y_j}]$ contains some event
  satisfying $\psi$. The formula $\Upsilon'$ identifies
  some $\pi_k$ such that $\minpie {\pi_k} {y_k}$ is maximal (first line),
  some $\pi_\ell$ such that $\maxpie {\pi_\ell} {y_\ell}$ is minimal
  (second line), and tests that there exists an event $x_n$ satisfying $\psi$
  between the two (third line).
  This is illustrated in Figure~\ref{fig:fo-pdl}.
  \begin{figure}[t]
    \centering
    \begin{tikzpicture}[font=\footnotesize,inner sep=2pt,auto,>=stealth]
      \node[dot,colj,label={[colj]below:{$y_j$}}] (xj) at (0,-1) {};
      \node[dot,colj] (xj1) at (-1.5,0) {};
      \node[dot,colj] (xj2) at (1,0) {};
      \path[colj] (xj) edge[decorate,decoration={snake,post length=0.5mm,amplitude=0.4mm},->=2] node[left,pos=0.4,xshift=-2pt] {$\minpi{\pi_j}$} (xj1)
      (xj) edge[decorate,decoration={snake,post length=0.5mm,amplitude=0.4mm},->=2] node[right,pos=0.4,xshift=5pt] {$\maxpi{\pi_j}$} (xj2) ;
      \draw[colj,dashed] (xj1) -- (xj2);
      
      \node[dot,coll,label={[coll]above:{$y_\ell$}}] (xl) at (-1.5,1) {};
      \node[dot,coll] (xl1) at (-1,0) {};
      \node[dot,coll] (xl2) at (0.5,0) {};
      \path[coll] (xl) edge[decorate,decoration={snake,post length=0.5mm,amplitude=0.4mm},->=2] node[left,pos=0.5] {$\minpi{\pi_\ell}$} (xl1)
      (xl) edge[decorate,decoration={snake,post length=0.5mm,amplitude=0.4mm},->=2] node[right,pos=0.1,xshift=5pt] {$\maxpi{\pi_\ell}$} (xl2) ;
      
      \node[dot,colk,label={[colk]above:{$y_k$}}] (xk) at (1.5,1) {};
      \node[dot,colk] (xk1) at (-0.5,0) {};
      \node[dot,colk] (xk2) at (2,0) {};
      \path[colk] (xk) edge[decorate,decoration={snake,post length=0.5mm,amplitude=0.4mm},->=2] node[left,pos=0.1,yshift=5pt] {$\minpi{\pi_k}$} (xk1)
      (xk) edge[decorate,decoration={snake,post length=0.5mm,amplitude=0.4mm},->=2] node[right] {$\maxpi{\pi_k}$} (xk2) ;

      \path[coll!50!colk] (xk1) edge[thick] node[below] {$x_n$} (xl2);
    \end{tikzpicture}
    \caption{Proof of Claim~\ref{claim:fo-pdl}}
    \label{fig:fo-pdl}
  \end{figure}

  We give now the formal proof of Claim~\ref{claim:fo-pdl}.
  Assume $M,\nu \models \Upsilon$.
  There exists $e \in E$ such that
  for all $j \in J$, $M,\nu(y_{j}),e \models \pi_{j}$, and
  for all $j \in J'$, $M,e \models \Loop {\pi_j}$.
  In particular, all $\minpie {\pi_{j}} {\nu(y_{j})}$ and
  $\maxpie {\pi_{j}} {\nu(y_{j})}$ for $j \in J$ are well-defined
  and on process $\proc(e)$.
  Let $k \in J$ such that $\minpie {\pi_{k}} {\nu(y_{k})}$ is
  maximal, i.e., $\minpie {\pi_{j}} {\nu(y_{j})} \leproc
  \minpie {\pi_{k}} {\nu(y_{k})}$ for all $j \in J$.
  Then, for all $j \in J$, we have $M,\nu(y_{j}),\nu(y_{k}) \models
  (\minpi {\pi_{j}}) \cdot {\righta} \cdot (\minpi {\pi_{k}})^{-1}$.
  Similarly,
  let $\ell \in J$ such that $\maxpie {\pi_{\ell}} {\nu(y_{\ell})}$
  is minimal. Then, for all $j \in J$, $M,\nu(y_{\ell}),\nu(y_{j})
  \models (\maxpi {\pi_{\ell}}) \cdot {\righta} \cdot
  (\maxpi {\pi_{j}})^{-1}$.
  In addition, we have $M,e \models \psi$,
  $M,\nu(y_{k}),e \models \pi_{k}$, and
  $M,e,\nu(y_{\ell}), \models \pi_{\ell}^{-1}$,
  hence $M,\nu(y_{k}),\nu(y_{\ell}) \models
  \pi_{k} \cdot \test{\psi} \cdot \pi_{\ell}^{-1}$.
  So we have 
  $M,\nu \models \Upsilon'$.
  
  Conversely, assume $M,\nu \models \Upsilon'$.
  Let $k,\ell \in J$ such that the corresponding sub-formula is satisfied.
  There exists $e \in E$ such that $M,\nu(y_{k}),e \models \pi_{k}$,
  $M,e \models \psi$, and $M,e,\nu(y_{\ell}) \models \pi_{\ell}^{-1}$.
  Note that we have
  $\minpie {\pi_{k}} {\nu(y_{k})} \leproc e \leproc
  \maxpie {\pi_{\ell}} {\nu(y_{\ell})}$.
  For all $j \in J'$, we have $M,e \models \Loop {\pi_{j}}$,
  i.e., $M,\nu[x \mapsto e] \models \pi_{j}(x_{n},x_n)$.
  Now, let $j \in J$. We have $M,\nu(y_{j}),\nu(y_{k}) \models
  (\minpi {\pi_{j}}) \cdot {\righta} \cdot (\minpi {\pi_{k}})^{-1}$,
  hence $\minpie {\pi_{j}} {\nu(y_{j})} \leproc
  \minpie {\pi_{k}} {\nu(y_{k})} \leproc e$.
  Similarly,
  $M,\nu(y_{\ell}),\nu(y_{j}) \models (\maxpi {\pi_{\ell}}) \cdot
  {\righta} \cdot (\maxpi {\pi_{j}})^{-1}$, hence
  $e \leproc \maxpie {\pi_{\ell}} {\nu(y_{\ell})} \leproc
  \maxpie {\pi_{j}} {\nu(y_{j})}$.
  In addition, since $M,e \models \psi$, we have
  $M,e \models \existsptrue{\pi_{j}^{-1}}$.
  Applying Lemma~\ref{lem:image}, we get
  $M,\nu(y_{j}),e \models \pi_{j}$, i.e.,
  $M,\nu[x \mapsto e] \models \pi_j(y_j,x_n)$.
  Hence, $M,\nu \models \Upsilon$.
  This concludes the proof of Claim~\ref{claim:fo-pdl}.

  Thus, $\Upsilon$ is equivalent to some positive
  combination of formulas $\pi(x,y)$ with $\pi \in \PDLm$ and 
  $x,y\in\{x_1,\ldots,x_{n-1}\}=\Free(\Phi)$, therefore,
  so is $\Phi$.
  Note that the two formulas $\bigl((\minpi {\pi_{j}}) \cdot {\righta} \cdot
  (\minpi {\pi_{k}})^{-1}\bigr) (y_{j},y_{k})$ and 
  $\bigl((\maxpi {\pi_{\ell}}) \cdot {\righta}
  \cdot (\maxpi {\pi_{j}})^{-1}\bigr) (y_{\ell},y_{j})$ are not $\PDLm$
  formulas (since $\righta$ is not). However, they are disjunctions of
  $\PDLm$ formulas, for instance, 
  $\bigl((\minpi {\pi_{j}}) \cdot {\righta} \cdot
  (\minpi {\pi_{k}})^{-1}\bigr) (y_{j},y_{k}) \equiv
  \bigl((\minpi {\pi_{j}}) \cdot (\minpi {\pi_{k}})^{-1}\bigr) (y_{j},y_{k})
  \lor \bigl((\minpi {\pi_{j}}) \cdot {\rightp} \cdot
  (\minpi {\pi_{k}})^{-1}\bigr) (y_{j},y_{k})$.
\end{proof}

We are now able to prove the main result relating $\FO[\prel,\mrel,\le]$ and
$\PDLm$.

\begin{proof}[Proof of Theorem~\ref{thm:FO-to-PDLm-main}]
  Let $\Phi_2(x_1,x_2)$ be an $\FO[\prel,\mrel,\le]$ formula with two free
  variables.  We apply Proposition~\ref{thm:FO-to-PDLm} to $\Phi_2(x_1,x_2)$ and
  obtain a boolean combination of path formulas $\pi(y,z)$ with
  $y,z\in\{x_1,x_2\}$.  First, we bring it into a positive
  boolean combination
  using Lemma~\ref{lem:closure-complement}.  Next, we replace formulas
  $\pi(x_1,x_1)$ with $\bigvee_{p,q} (\test{\Loop{\pi}}\cdot\jump{p}{q})(x_1,x_2)$.
  Similarly, $\pi(x_2,x_2)$ is replaced with $\bigvee_{p,q}
  (\jump{p}{q}\cdot\test{\Loop{\pi}})(x_1,x_2)$.  Also, $\pi(x_2,x_1)$ is replaced with
  $\pi^{-1}(x_1,x_2)$.  Finally, we transform it into disjunctive normal form:
  we obtain $\Phi_1(x_1,x_2)\equiv\bigvee_{i}\bigwedge_{j} 
  \pi_{ij}(x_1,x_2)$, which concludes the proof in the case of two free 
  variables.

    Next, let $\Phi_1(x)$ be an $\FO[\prel,\mrel,\le]$ formula with one free
    variable.
    As above, applying Proposition~\ref{thm:FO-to-PDLm} to $\Phi_1(x)$
    and then Lemma~\ref{lem:closure-complement}, we obtain $\PDLm$ path
    formulas $\pi_{ij}$ such that
    $\Phi_1(x) \equiv  \bigvee_i \bigwedge_j \pi_{ij}(x,x)$.
    Now, $M,[x \mapsto e] \models \pi_{ij}(x,x)$ iff
    $M,e \models \Loop{\pi_{ij}}$.
    Hence, $\Phi(x) \equiv \bigvee_i \bigwedge_j \Loop{\pi_{ij}}$.
  
  Finally, an $\FO[\prel,\mrel,\le]$ sentence $\Phi_0$ is a boolean combination of 
  formulas of the form $\exists x.\Phi_1(x)$. Applying the theorem to $\Phi_1(x)$, we 
  obtain an equivalent $\PDLm$ event formula $\varphi$. Then, we take 
  $\pdlsentence=\E\varphi$, which is trivially equivalent to $\exists x.\Phi_1(x)$.
\end{proof}

\section{From $\boldsymbol{\sfPDLm}$ to CFMs}\label{sec:pdl-cfm}

In the inductive translation of $\sfPDLm$ formulas into CFMs,
\emph{event} formulas will be evaluated by \emph{MSC transducers}.
An MSC transducer for a formula $\varphi$ produces a truth value at every event
on the given MSC. More precisely, it outputs $1$ when $\varphi$ holds, and $0$ otherwise.
We will first introduce MSC transducers formally and then go into the actual translation.

\subsection{Letter-to-letter MSC Transducers}

Let $\Gamma$ be a nonempty finite output alphabet.
A \emph{(nondeterministic) letter-to-letter MSC transducer} (or simply, \emph{transducer})
$\A$ over $\Procs$ and from $\Sigma$ to $\Gamma$ is a CFM
over $\Procs$ and $\Sigma \times \Gamma$.
The transducer $\A$ accepts the relation
$\Lt \A = \{\bigl((E,\prel,\mrel,\ploc,\lambda),(E,\prel,\mrel,\ploc,\gamma)\bigr)
\mid {(E,\prel,\mrel,\ploc,\lambda\times\gamma)} \in \msclang(\A)\}$.
Transducers are closed under product and composition, using standard
constructions:

\begin{lemma}
  Let $\A$ be a transducer from $\Sigma$ to $\Gamma$, and $\A'$ a transducer
  from $\Sigma$ to $\Gamma'$.
  There exists a transducer $\A \times \A'$ from $\Sigma$ to $\Gamma \times \Gamma'$ such that
  \begin{align*}
    \Lt {\A \times \A'}
    & = \big\{
      \big((E,\prel,\mrel,\ploc,\lambda),
      (E,\prel,\mrel,\ploc,\gamma \times \gamma')\big)
      \mid {} \\
    & \hspace{5em}
      \big((E,\prel,\mrel,\ploc,\lambda),(E,\prel,\mrel,\ploc,\gamma)\big)
      \in \Lt {\A}, \\
    & \hspace{5em}
      \big((E,\prel,\mrel,\ploc,\lambda),(E,\prel,\mrel,\ploc,\gamma')\big)
      \in \Lt {\A'}
      \big\} \, .
  \end{align*}
\end{lemma}

\begin{lemma}
  Let $\A$ be a transducer from $\Sigma$ to $\Gamma$, and $\A'$ a
  transducer from $\Gamma$ to $\Gamma'$.
  There exists a transducer $\A' \circ \A$ from $\Sigma$ to $\Gamma'$ such that
  $$
    \Lt {\A' \circ \A} = \Lt {\A'} \circ \Lt {\A} =
    \{ (M,M'') \mid \exists M' \in \MSCs{\Procs}{\Gamma} :
    (M,M') \in \Lt {\A}, (M',M'') \in \Lt {\A'} \} \, .
  $$
\end{lemma}

\subsection{Translation of $\boldsymbol{\sfPDLm}$ Event Formulas into CFMs}

For a $\PDLm$ event formula $\varphi$ and an MSC
$M = (E,\prel,\mrel,\ploc,\lambda)$ over $\Procs$ and $\Sigma$,
we define an MSC $\Mphi M \varphi = (E,\prel,\mrel,\ploc,\gamma)$ over
$\Procs$ and $\{0,1\}$, by setting $\gamma(e) = 1$ if $M,e \models \varphi$,
and $\gamma(e) = 0$ otherwise.
Our goal is to construct a transducer $\Aphi \varphi$ such that
$\Lt {\Aphi \varphi} = \{(M,\Mphi M \varphi) \mid M \in \MSCs \Procs \Sigma\}$.

We start with the case of formulas from $\sfPDL[\emptyset]$, i.e.,  without $\Loopop$.

\begin{lemma}\label{lem:trad-loop-free}
  Let $\varphi$ be a $\sfPDL[\emptyset]$ event formula.
  There exists a transducer $\Aphi \varphi$
  such that
  $\Lt {\Aphi \varphi} = \{(M,\Mphi M \varphi) \mid M \in \MSCs \Procs \Sigma\}$.
\end{lemma}

\begin{proof}
  Any $\PDL[\emptyset]$ event formula is equivalent to some formula
  $\varphi$ over the syntax
  $$
    \varphi ::= p \mid a \mid \varphi \lor \varphi \mid \lnot \varphi \mid
    \existsp{\mrel_{p,q}}{\varphi} \mid \existsp{\mrel_{p,q}^{-1}}{\varphi}
    \mid \existsp{\rightg \varphi}{\varphi}
    \mid \existsp{\leftg \varphi}{\varphi}
    \mid \existsp{\jump p q}{\varphi}
  $$
  Indeed, we have $\existsp{\pi_1 \cdot \pi_2}{\varphi} \equiv
  \existsp{\pi_1}{(\existsp{\pi_2}{\varphi})}$,
  and $\existsp{\test{\varphi}}{\psi} \equiv \varphi \land \psi$.
  Notice that ${\prel}\equiv{\rightg{\False}}$ and ${\leftmove}\equiv{\leftg{\False}}$.

  It is easy to define $\A_\varphi$ for formulas $\varphi = p$, with $p \in P$,
  or $\varphi = a$, with $a \in \Sigma$.
  We also use below simple transducers over
  $P$ from $\{0,1\}^2$ or $\{0,1\}$ to $\{0,1\}$.
  For instance, the transducer $\B_{\neg}$ from $\{0,1\}$ to $\{0,1\}$ outputs 
  the negation of the bit read and $\B_{\vee}$ from $\{0,1\}^{2}$ to 
  $\{0,1\}$ outputs the disjunction of the two bits read.
  The transducer $\B_{\mrel_{p,q}}$ from $\{0,1\}$ to $\{0,1\}$ 
  outputs 1 at an event $e$ iff $e$ is a send
  event from $p$ to $q$ and the corresponding receive event $f$ is labeled 1.
  To do so, at each send event $e$ from $p$ to $q$, the transducer guesses
  whether the corresponding receive event $f$ is labeled $0$ or 1, outputs its
  guess and sends it on the message from $e$ to $f$.  At the receive event $f$
  the transducer checks that the guess was correct.  The run is accepting if all
  guesses were correct.
  The deterministic transducer $\B_{\mathsf{YS}}$ from $\{0,1\}^{2}$ to
  $\{0,1\}$ corresponds to the \emph{strict since} modality.  On each process,
  it runs the automaton given in Figure~\ref{fig:YS}:
  it outputs 1 at some event
  $e$ if there is $g\ltproc e$ where the second bit is 1 and for all $g\ltproc
  f\ltproc e$ the first bit at $f$ is 1.
  Similarly, we can construct the nondeterministic transducer 
  $\B_{\mathsf{XU}}$ for the \emph{strict until}.
  Finally, it is easy to construct a transducer $\B_{\jump p q}$ which outputs 0 on
  all events of processes $r\neq p$ and outputs 1 (resp.\ 0) on all events of
  process $p$ iff some event (resp.\ no event) of process $q$ is labeled 1.
  \begin{figure}[tbp]
      \centering
    \includegraphics[page=1]{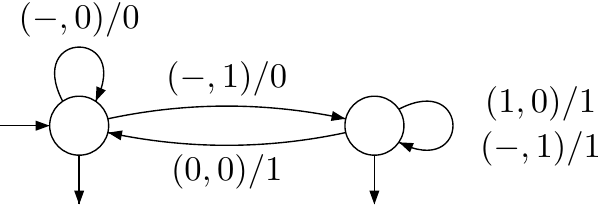}
    \caption{Transducer for strict since.
      In a transition $a/b$, the input is $a$ and the output is $b$.
      Write and receive actions are omitted.}
    \label{fig:YS}
  \end{figure}
  We then let
  \begin{align*}
    \A_{\varphi_1 \lor \varphi_2}
    & = \B_{\lor} \circ (\A_{\varphi_1} \times \A_{\varphi_2})
    & \A_{\lnot \varphi}
    & = \B_{\neg} \circ \A_{\varphi} \\
    \A_{\existsp{\mrel_{p,q}}{\varphi}}
    & = \B_{\mrel_{p,q}} \circ \A_{\varphi}
    & \A_{\existsp{\mrel_{p,q}^{-1}}{\varphi}}
    & = \B_{\mrel_{p,q}^{-1}} \circ \A_{\varphi} \\
    \A_{\existsp{\rightg {\varphi_1}}{\varphi_2}}
    & = \B_{\mathsf{XU}}
      \circ (\A_{\varphi_1} \times \A_{\varphi_2})
    & \A_{\existsp{\jump p q}{\varphi}}
    & = \B_{\jump p q} \circ \A_{\varphi}\\
    \A_{\existsp{\leftg {\varphi_1}}{\varphi_2}}
    & = \B_{\mathsf{YS}}
      \circ (\A_{\varphi_1} \times \A_{\varphi_2}) \, . && \qedhere
  \end{align*}
\end{proof}

Next, we look at a single loop where the path $\pi \in \sfPDL[\emptyset]$ is of
the form $\minpi{\pi'}$ or $\maxpi{\pi'}$.  This case will be simpler than
general loop formulas, because of the fact that $\semM {\minpi {\pi'}}(e)$ is
always either empty or a singleton.  Recall that, in addition, $\minpi{\pi'}$ is
monotone.

\begin{lemma}\label{lem:det-loop}
  Let $\pi$ be a $\sfPDLall[\emptyset]$ path formula of the form $\pi = \minpi {\pi'}$
  or $\pi = \maxpi {\pi'}$, and let $\varphi = \Loop \pi$.
  There exists a transducer $\Aphi {\varphi}$
  such that $\Lt {\Aphi {\varphi}} =
  \{(M,\Mphi M {\varphi}) \mid M \in \MSCs \Procs \Sigma\}$.
\end{lemma}

\begin{proof}
We can assume that $\Comp{\pi} \subseteq \id$.
  We define $\Aphi{\varphi}$ as the composition of three transducers that will
  guess and check the evaluation of $\varphi$.  More precisely, $\Aphi{\varphi}$
  will be obtained as an inverse projection $\alpha^{-1}$, followed by the
  intersection with an MSC language $K$, followed by a projection $\beta$.

  We first enrich the labeling of the MSC with a color from $\Theta =
  \{\ycolone,\ycoltwo,\ncolone,\ncoltwo\}$.  Intuitively, colors $\ycolone$ and
  $\ycoltwo$ will correspond to a guess that the formula $\varphi$ is satisfied,
  and colors $\ncolone$ and $\ncoltwo$ to a guess that the formula is not
  satisfied. Consider the projection 
  $\alpha\colon\MSCs{\Procs}{\Sigma\times\Theta}\to\MSCs{\Procs}{\Sigma}$ which erases 
  the color from the labeling. The inverse projection $\alpha^{-1}$ can be 
  realized with a transducer $\A$, i.e.,
  $\Lt{\A}=\{(\alpha(M'),M')\mid M'\in\MSCs{\Procs}{\Sigma\times\Theta}\}$.
  
  Define the projection
  $\beta\colon\MSCs{\Procs}{\Sigma\times\Theta}\to\MSCs{\Procs}{\{0,1\}}$ by
  $\beta((E,\prel,\mrel,\ploc,\lambda\times\theta))=(E,\prel,\mrel,\ploc,\gamma)$,
  where $\gamma(e)=1$ if $\theta(e)\in\{\ycolone,\ycoltwo\}$, and $\gamma(e)=0$
  otherwise. The projection $\beta$ can be realized with a transducer $\A''$: 
  we have
  $\Lt{\A''}=\{(M',\beta(M'))\mid M'\in\MSCs{\Procs}{\Sigma\times\Theta}\}$.
  
  Finally, consider the language $K\subseteq\MSCs{\Procs}{\Sigma\times\Theta}$ of 
  MSCs $M'=(E,\prel,\mrel,\ploc,\lambda\times\theta)$ satisfying the following 
  two conditions:
  \begin{enumerate}[nosep]
    \item Colors $\ycolone$ and $\ycoltwo$ alternate on each process $p\in \Procs$:
    if $e_1 < \cdots < e_n$ are the events in $E_p \cap
    \theta^{-1}(\{\ycolone,\ycoltwo\})$, then $\theta(e_i) = \ycolone$ if $i$ is
    odd, and $\theta(e_i) = \ycoltwo$ if $i$ is even.
  
    \item For all $e \in E$, $\theta(e)\in\{\ycolone,\ycoltwo\}$ iff there
    exists $f \in E$ such that $M,e,f \models \pi$ and $\theta(e) = \theta(f)$.
  \end{enumerate}
  The first property is trivial to check with a CFM. Using 
  Lemma~\ref{lem:trad-loop-free}, we show
  that the second property can also be checked with a CFM.
  First, from $\pi$ we construct a $\sfPDLall[\emptyset]$ event formula
  $\psi$ over $\Procs$ and $\Sigma \times \Theta$ such that, for all
  $M'= (E,\prel,\mrel,\ploc,\lambda\times\theta)\in\MSCs{\Procs}{\Sigma\times\Theta}$
  and events $e \in E$, we have $M',e \models \psi$ iff the following holds:
  $\theta(e)\in\{\ycolone,\ycoltwo\}$ iff there is $f \in E$ such that $\alpha(M'),e,f \models \pi$ and
  $\theta(e) = \theta(f)$.
  Namely, we define
  $$\psi = (\ycolone \;\vee\; \ycoltwo) \Longleftrightarrow
  \Bigl[\bigl( {\ycolone} \;\wedge \existsp{\hat\pi}{\ycolone}\bigr) \vee
  \bigl( {\ycoltwo} \;\wedge \existsp{\hat\pi}{\ycoltwo}\bigr)\Bigr]$$
  where the state formula $\col \in \{\ycolone,\ycoltwo\}$ is an abbreviation for
  $\bigvee_{a \in \Sigma} (a,\col)$
  and $\hat\pi$ is obtained from $\pi$
  by replacing state formulas $a$ with $\bigvee_{\col \in \Theta} (a,\col)$.
  Now, the language for the second condition is
  $\{M' \in\MSCs{\Procs}{\Sigma\times\Theta} \mid$ every event of $M'_\psi$ is labeled with $1\}$,
  for which we can easily give a CFM using the transducer $\Aphi{\psi}$ from $\Sigma \times \Theta$ to $\{0,1\}$ given by
  Lemma~\ref{lem:trad-loop-free}.
  
  We deduce that there is a transducer $\A'$ 
  such that $\Lt{\A'}=\{(M',M')\mid M'\in K\}$.
  We let $\Aphi{\varphi} = \A''\circ\A'\circ\A$. Notice that
  $\Lt{\Aphi{\varphi}}=\{(\alpha(M'),\beta(M'))\mid M'\in K\}$.
  From the following two claims, we deduce immediately that 
  $\Lt{\Aphi{\varphi}}=\{(M,M_\varphi)\mid M\in\MSCs{\Procs}{\Sigma}\}$.
  
  \begin{claim}\label{cl:forest}
    For all $M\in\MSCs{\Procs}{\Sigma}$, there exists $M'\in K$ with $\alpha(M')=M$.
  \end{claim}
  
  \begin{proof}[Proof of Claim~\ref{cl:forest}]
  Let $M=(E,\prel,\mrel,\ploc,\lambda)\in\MSCs{\Procs}{\Sigma}$.
  Let $E_1=\{e \in E \mid M,e \models \varphi\}$ and $E_0=E\setminus E_1$.
  Consider the graph $G = (E, \{(e,f) \mid M,e,f \models \pi\})$.
  Since $\pi = \minpi {\pi'}$ or $\pi = \maxpi {\pi'}$, every vertex has
  outdegree at most 1, and, by Lemma~\ref{lem:monotone},
  there are no cycles except for self-loops. So the restriction of $G$ to
  $E_0$ is a forest, and there exists a $2$-coloring
  $\chi \colon E_0 \to \{\ncolone,\ncoltwo\}$ such that, for all $e,f \in E_0$
  with $M,e,f \models \pi$, we have $\chi(e) \neq \chi(f)$.
  This is illustrated in Figure~\ref{fig:Gpi}.
  Moreover, there exists $\theta \colon E \to \Theta$ such that $\theta(e) = \chi(e)$ for
  $e \in E_0$, and $\theta(e)\in\{\ycolone,\ycoltwo\}$ for $e \in E_1$ is 
  such that Condition 1 of the definition of $K$ is satisfied. It is easy to 
  see that Condition 2 is also satisfied. 
  Indeed, if 
  $\theta(e)\in\{\ycolone,\ycoltwo\}$, then $e\in E_1$ and $M,e,e\models\pi$.
  Now, if $\theta(e)\notin\{\ycolone,\ycoltwo\}$, then $e\in E_0$ and either 
  $M,e\not\models\existsptrue{\pi}$ or, by definition of $\theta$, we have 
  $\theta(e)\neq\theta(f)$ for the unique $f$ such that $M,e,f\models\pi$.
  \qedhere(Claim~\ref{cl:forest})
  \end{proof}
  
  \tikzstyle{dot} = [white, draw=black, line width=0.5pt, fill, circle, inner sep=0, minimum size = 6pt]
\tikzstyle{bdot} = [gray, draw=black, line width=0.5pt, circle, fill, inner sep=0, minimum size = 6pt]
\tikzstyle{rdot} = [white, draw=black, line width=0.5pt, fill, rectangle, inner sep=0, minimum size = 6pt]
\tikzstyle{rbdot} = [gray, draw=black, line width=0.5pt, rectangle, fill, inner sep=0, minimum size = 6pt]

   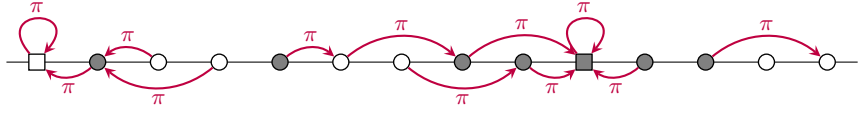
\begin{figure}
   \centering
    \begin{tikzpicture}[auto,>=stealth, scale=0.8,
      every loop/.style={in=50,out=130,min distance=1cm}, bend angle=40,
      inner sep=2pt,font=\footnotesize]
      \draw (0.5,0) -- (14.5,0);
      \foreach \i in {3,4,6,7,13,14} {
        \node[dot] (e\i) at (\i,0) {};
      }
      \foreach \i in {2,5,8,9,11,12} {
        \node[bdot] (e\i) at (\i,0) {};
      }
      \node[rdot] (e1) at (1,0) {};
      \node[rbdot] (e10) at (10,0) {};

      \path[purple,thick,->,font=\footnotesize]
      (e1) edge[loop above] node {$\pi$} (e1)
      (e10) edge[loop above] node {$\pi$} (e10)
      (e2) edge[bend left] node {$\pi$} (e1)
      (e4) edge[bend left] node {$\pi$} (e2)
      (e3) edge[bend right] node[above] {$\pi$} (e2)
      (e11) edge[bend left] node {$\pi$} (e10)
      (e9) edge[bend right] node[below] {$\pi$} (e10)
      (e7) edge[bend right] node[below] {$\pi$} (e9)
      (e5) edge[bend left] node {$\pi$} (e6)
      (e6) edge[bend left] node {$\pi$} (e8)
      (e8) edge[bend left] node {$\pi$} (e10)
      (e12) edge[bend left] node {$\pi$} (e14);
    \end{tikzpicture}
    \caption{Proof of Claim~\ref{cl:forest}\label{fig:Gpi}: $2$-coloring of $E_0$ in Graph $G$}
  \end{figure}

  \begin{claim}\label{cl:beta}
    For all $M'\in K$, we have $\beta(M')=M_\varphi$, where $M=\alpha(M')$.
  \end{claim}
  
  \begin{proof}[Proof of Claim~\ref{cl:beta}]
  Let $M'=(E,\prel,\mrel,\ploc,\lambda\times\theta)\in K$ and $M=\alpha(M')$.
  Suppose towards a contradiction that
  $M_\varphi\neq\beta(M)=(E,\prel,\mrel,\ploc,\gamma)$.  By Condition~2, for all
  $e \in E$ such that $\gamma(e) = 0$, we have $M,e \not\models \varphi$.  So
  there exists $f_0 \in E$ such that $\gamma(f_0) = 1$ and $M,f_0 \not\models
  \varphi$.  Notice that $\theta(f_0)\in\{\ycolone,\ycoltwo\}$.  For all $i \in
  \Nat$, let $f_{i+1}$ be the unique event such that $M,f_i,f_{i+1} \models
  \pi$.  Such an event exists by Condition~2, and is unique since $\pi = \minpi
  {\pi'}$ or $\pi = \maxpi {\pi'}$.  Note that, for all $i$, $\theta(f_{i+1}) =
  \theta(f_i) \in \{\ycolone,\ycoltwo\}$.  Suppose $f_0 \ltproc f_1$ (the case
  $f_1 \ltproc f_0$ is similar).  By Condition~1, there exists $g_0$ such that
  $f_0 \ltproc g_0 \ltproc f_1$ and
  $\{\theta(f_0),\theta(g_0)\}=\{\ycolone,\ycoltwo\}$.
  For an illustration, see Figure~\ref{fig:beta}.
  Again, for all $i \in
  \Nat$, let $g_{i+1}$ be the unique event such that $M,g_i,g_{i+1} \models
  \pi$.  Note that all $f_0,f_1,\ldots$ have the same color, in
  $\{\ycolone,\ycoltwo\}$, and all $g_0,g_1,\ldots$ carry the complementary
  color.  Thus, $f_i \neq g_j$ for all $i,j \in \Nat$.  But, by
  Lemma~\ref{lem:monotone}, this implies $f_0 \ltproc g_0
  \ltproc f_1 \ltproc g_1 \ltproc \cdots$, which contradicts the fact that we
  deal with finite MSCs.
  \qedhere(Claim~\ref{cl:beta})
\end{proof}

\begin{figure}
   \centering
    \begin{tikzpicture}[auto,>=stealth, scale=0.8,
      every loop/.style={in=50,out=130,min distance=1cm}, bend angle=40,
      inner sep=2pt,font=\footnotesize]
      \draw (0.5,0) -- (14.5,0);
      \foreach \i in {1,3,5,7,9,11,13} {
        \node[rdot] (e\i) at (\i,0) {};
      }
      \foreach \i in {2,4,6,8,10,12,14} {
        \node[rbdot] (e\i) at (\i,0) {};
      }
      \node (inf) at (15,0.8) {$\textcolor{red}\cdots$};
      
     \node (f0) at (1,-0.5) {$f_0$};
     \node (f1) at (5,-0.5) {$f_1$};
     \node (f2) at (11,-0.5) {$f_2$};
     \node (g0) at (2,-0.5) {$g_0$};
     \node (g1) at (8,-0.5) {$g_1$};
     \node (g2) at (14,-0.5) {$g_2$};

      \path[purple,thick,->,font=\footnotesize]
      (e1) edge[bend left] node {$\pi$} (e5)
      (e2) edge[bend left] node {$\pi$} (e8)
      (e5) edge[bend left] node {$\pi$} (e11)
      (e8) edge[bend left] node {$\pi$} (e14)
      (e11) edge[bend left] node {$\pi$} (inf);
    \end{tikzpicture}
    \caption{Proof of Claim~\ref{cl:beta}\label{fig:beta}}
  \end{figure}

This concludes the proof of Lemma~\ref{lem:det-loop}.
\end{proof}

  The general case is more complicated. We first show how to rewrite an arbitrary loop 
  formula using loops on paths of the form $\maxpi{\pi}$ or 
  $(\maxpi{\pi})\cdot{\leftp}$.
  Intuitively, this means that loop formulas will only be used to test,
  given an event $e$ such that $e' = \maxpie \pi e$ is well-defined and
  on the same process as $e$, whether $e' \ltproc e$, 
  $e' = e$, or $e \ltproc e'$.
  Indeed, we have $M,e \models \Loop{(\maxpi{\pi})\cdot{\leftp}}$ iff
  $e \ltproc \maxpie \pi e$.

\begin{lemma}\label{lem:loops}
  For all $\PDLm$ path formulas $\pi$,
  $$\Loop \pi
    \equiv \Loop{\maxpi \pi} \lor
      \left({\existsptrue{\pi^{-1}}} \land {\Loop{(\maxpi \pi) \cdot {\leftp}}}
      \land {\lnot \Loop{(\minpi \pi) \cdot {\leftp}}}\right)
      \, .
$$
\end{lemma}

\begin{proof}
  The result follows from Lemma~\ref{lem:image}.
  Indeed, if we have $M,e \models \Loop \pi$ and $M,e \not\models \Loop {\maxpi \pi}$,
  then $\minpie \pi e \leproc e \ltproc \maxpie \pi e$
  and $M,e \models \existsptrue{\pi^{-1}}$, hence
  $M,e \models {\existsptrue{\pi^{-1}}}
  \land {\Loop{(\maxpi \pi) \cdot {\leftp}}}
  \land {\lnot \Loop{(\minpi \pi) \cdot {\leftp}}}$.
  Conversely, if $M,e \models \Loop {\maxpi \pi}$, then
  $M,e \models \Loop \pi$, and if $M,e \models ({\existsptrue{\pi^{-1}}}
  \land {\Loop{(\maxpi \pi) \cdot {\leftp}}}
  \land {\lnot \Loop{(\minpi \pi) \cdot {\leftp}}})$,
  then $M,e \models \existsptrue{\pi^{-1}}$ and $\minpie \pi e \leproc e
  \ltproc \maxpie \pi e$, hence $M,e,e \models \pi$, i.e., $M,e \models
  \Loop \pi$.
\end{proof}

Notice that, since $\minpi \pi \equiv \maxpi {(\minpi \pi)}$,
the formula $\Loop{(\minpi \pi) \cdot {\leftp}}$ can also be seen as
a special case of a $\Loop{(\maxpi {\pi'}) \cdot {\leftp}}$ formula.

\begin{theorem}\label{thm:PDLp-to-CFM}
  For all $\PDLm$ event formulas $\varphi$, there exists a transducer
  $\Aphi \varphi$
  such that
  $\Lt {\Aphi \varphi} = \{(M,\Mphi M \varphi) \mid M \in \MSCs \Procs \Sigma\}$.
\end{theorem}

\begin{proof}
  By Lemma~\ref{lem:loops}, we can assume that all loop subformulas in $\varphi$
  are of the form $\Loop {(\maxpi{\pi})\cdot{\leftp}}$ or $\Loop{\maxpi{\pi}}$
  (recall that $\minpi \pi \equiv \maxpi {(\minpi \pi)}$).  We prove
  Theorem~\ref{thm:PDLp-to-CFM} by induction on the number of loop subformulas
  in $\varphi$.  The base case is stated in Lemma~\ref{lem:trad-loop-free}.

  Let $\psi = \Loop{\pi'}$ be a subformula of $\varphi$ such that $\pi'$
  contains no loop subformulas and $\Comp{\pi'} \subseteq \id$.
  Let us show that there exists $\Aphi \psi$ such that
  $\Lt {\Aphi \psi} = \{(M,\Mphi M \psi) \mid M \in \MSCs \Procs \Sigma\}$.
  If $\pi'=\maxpi{\pi}$, then we apply Lemma~\ref{lem:det-loop}.
  Otherwise, $\pi'=(\maxpi{\pi})\cdot{\leftp}$ for some $\sfPDLall[\emptyset]$ 
  path formula $\pi$. So we assume from now on that 
  $\psi=\Loop{(\maxpi{\pi})\cdot{\leftp}}$.
  
  We start with some easy remarks.  Let $p\in\Procs$ be some process
  and $e\in E_p$.  A necessary condition for
  $M,e\models\psi$ is that
  $M,e\models\existsptrue{\pi}\wedge\neg\Loop{\maxpi{\pi}}$.  
  Also, it is easy to see that
  $M,e\models\Loop{\minpi{({\rightp}\cdot\pi^{-1})}}$ is a sufficient condition
  for $M,e\models\psi$.
  
  We let $E_p^{\pi}$
  be the set of events $e\in E_p$ satisfying $\existsp{\pi}{p}$. For all $e\in
  E_p^{\pi}$, we let $e'=\semM{\maxpi{\pi}}(e)\in E_p$. The transducer $\A_\psi$ will 
  establish, for each $e\in E_p^{\pi}$, whether $e'\ltproc e$, $e'=e$, or $e \ltproc e'$, and it 
  will output $1$ if $e \ltproc e'$, and $0$ otherwise. The case $e'=e$ means 
  $M,e\models\Loop{\maxpi{\pi}}$ and can be checked with the help of 
  Lemma~\ref{lem:det-loop}. So the difficulty is to distinguish between $e' \ltproc e$ 
  and $e \ltproc e'$ when $M,e\models\existsptrue{\pi}\wedge\neg\Loop{\maxpi{\pi}}$.

    The following two claims rely on Lemma~\ref{lem:monotone}.
    Recall that $\psi=\Loop{(\maxpi{\pi})\cdot{\leftp}}$.
  
  \begin{claim}\label{claim:minimal}
    Let $f$ be the minimal event in $E_p^{\pi}$ (assuming this set is 
    nonempty). Then, $M,f\models\psi$ iff 
    $M,f\models\Loop{\minpi{({\rightp}\cdot\pi^{-1})}}$.
  \end{claim}

  \begin{proof}[Proof of Claim~\ref{claim:minimal}]
    The right to left implication holds without any hypothesis. 
    Conversely, assume $f\rightp f'=\semM{\maxpi{\pi}}(f)$.
    Then, $M,f,f \models {\rightp}\cdot\pi^{-1}$,
    and $g = \semM{\minpi{({\rightp}\cdot\pi^{-1})}}(f) \leproc f$.
    This is illustrated in Figure~\ref{fig:minimal}.
    Moreover, $M,g \models \existsptrue{\pi}$ and by minimality of $f$ in
    $E_p^{\pi}$, we conclude that $g = f$.
    \qedhere(Claim~\ref{claim:minimal})
  \end{proof}

\begin{figure}[h]
   \centering
    \begin{tikzpicture}[auto,>=stealth, scale=0.8,
      every loop/.style={in=50,out=130,min distance=1cm}, bend angle=40,
      inner sep=2pt,font=\footnotesize]
    
     \node (f) at (1,0) {$f$};
     \node (fp) at (5,0) {$f'$};
     \node (g) at (0,0) {$g$};
     \node (eq) at (0.5,0) {$=$};

      \path
      (1.2,-0.1) edge[purple,thick,font=\footnotesize] (3,-0.1)
      (3,-0.1) edge[bend left,purple,thick,->,font=\footnotesize] node {$\minpi{({\rightp}\cdot\pi^{-1})}$} (g);

      \path
      (f) edge[->] node {$+$} (fp);

      \path[purple,thick,->,font=\footnotesize]
      (f) edge[bend left] node {$\max \pi$} (fp);
    \end{tikzpicture}
    \caption{Proof of Claim~\ref{claim:minimal}\label{fig:minimal}}
  \end{figure}

  \begin{claim}\label{claim:consecutive}
    Let $e,f$ be consecutive events in $E_p^{\pi}$, i.e., $e,f\in E_p^{\pi}$ 
    and $M,e,f\models{\rightg{\neg\existsptrue{\pi}}}$.
    \begin{enumerate}[nosep]
      \item
      If $M,e\not\models\psi$, then $[M,f\models\psi$ iff 
      $M,f\models\Loop{\minpi{({\rightp}\cdot\pi^{-1})}}]$.

      \item
      If $M,e\models\psi$, then $[M,f\not\models\psi$ iff 
      $M,f\models\Loop{\maxpi{\pi}}\vee
      \Loop{\maxpi{((\maxpi{\pi})\cdot{\rightg{\neg\existsptrue{\pi}})}}}]$.
    
    \end{enumerate}
  \end{claim}
  
  \begin{proof}[Proof of Claim~\ref{claim:consecutive}]
    We show the two statements.  
  \begin{enumerate}
    \item Assume that $M,e\not\models\psi$.
    Again, the right to left implication holds without any hypothesis. 
    Conversely, assume that  $M,e\not\models\psi$ and  
    $M,f\models\psi$, i.e., $e'=\semM{\maxpi{\pi}}(e)\leproc e$ and $f \ltproc 
    f'=\semM{\maxpi{\pi}}(f)$.
    We have $M,f\models\existsptrue{{\rightp}\cdot\pi^{-1}}$, and
    $g=\semM{\minpi{({\rightp}\cdot\pi^{-1})}}(f) \leproc f$.
    Notice that $g\in E_p^{\pi}$, and $f \ltproc g'=\semM {\maxpi \pi} (g)$.
    If $g \ltproc f$, we get $g\leproc e$, and
    using Lemma~\ref{lem:monotone} (monotonicity), we obtain
    $g' \leproc e' \leproc e \ltproc f$,
    a contradiction. The situation is illustrated in Figure~\ref{fig:consecutive1}.
    Therefore, $g=f$ and $M,f\models\Loop{\minpi{({\rightp}\cdot\pi^{-1})}}$.
    
   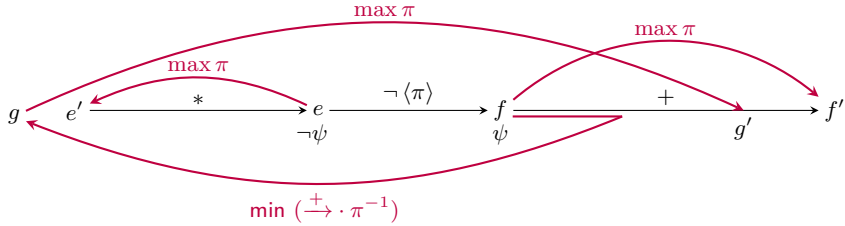
\begin{figure}[h]
   \centering
    \begin{tikzpicture}[auto,>=stealth, scale=0.8,
      every loop/.style={in=50,out=130,min distance=1cm}, bend angle=40,
      inner sep=2pt,font=\footnotesize]
    
     \node (f) at (1,0) {$f$};
     \node (fp) at (6.5,0) {$f'$};
     \node (e) at (-2,0) {$e$};
     \node (ep) at (-6,0) {$e'$};

     \node (notpsi) at (-2.1,-0.4) {$\neg\psi$};
     \node (psi) at (1,-0.4) {$\psi$};
     \node (g) at (-7,-0.1) {$g$};
     \node (gp) at (5,-0.3) {$g'$};

      \path
      (f) edge[->] node {$+$} (fp);

      \path
      (e) edge[->] node {$\neg\existsptrue{\pi}$} (f);

      \path
      (ep) edge[->] node {$\ast$} (e);

      \path[purple,thick,->,font=\footnotesize]
      (f) edge[bend left] node {$\max \pi$} (fp);
      
      \path[purple,thick,->,font=\footnotesize]
      (e) edge[bend right=25] node[above] {$\max \pi$} (ep);

      \path[purple,thick,->,font=\footnotesize]
      (g) edge[bend left=25] node {$\max \pi$} (5,0);

      \path
      (1.2,-0.1) edge[purple,thick,font=\footnotesize] (3,-0.1)
      (3,-0.1) edge[bend left=22,purple,thick,->,font=\footnotesize] node {$\minpi{({\rightp}\cdot\pi^{-1})}$} (g);

    \end{tikzpicture}
    \caption{Proof of Claim~\ref{claim:consecutive}(1.)\label{fig:consecutive1}}
  \end{figure}

    \item  Assume that $M,e\models\psi$.
    The right to left implication holds easily: the first disjunct implies that 
    $f'=f$ and the second disjunct implies $f'<f$.  Conversely, assume that
    $M,e\models\psi$, $M,f\not\models\psi$ and
    $M,f\not\models\Loop{\maxpi{\pi}}$, i.e., $e\ltproc e'$ and $f' \ltproc f$.
    From Lemma~\ref{lem:monotone} we get $e'\leproc f'$ and since 
    $e,f$ are consecutive in $E_p^{\pi}$ we obtain 
    $M,f',f\models{\rightg{\neg\existsptrue{\pi}}}$. Therefore,
    $M,f\models\Loop{(\maxpi{\pi})\cdot\rightg{\neg\existsptrue{\pi}}}
    \equiv\Loop{\maxpi{((\maxpi{\pi})\cdot\rightg{\neg\existsptrue{\pi}})}}$.

   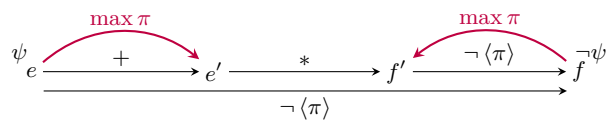
\begin{figure}[h]
   \centering
    \begin{tikzpicture}[auto,>=stealth, scale=0.8,
      every loop/.style={in=50,out=130,min distance=1cm}, bend angle=40,
      inner sep=2pt,font=\footnotesize]
    
     \node (e) at (1,0.3) {$e$};
     \node (ep) at (4,0.3) {$e'$};
     \node (fp) at (7,0.3) {$f'$};
     \node (f) at (10,0.3) {$f$};

     \node (a) at (1.1,-0.02) {};
     \node (b) at (9.9,-0.02) {};

     \node (notpsi) at (0.8,0.6) {$\psi$};
     \node (psi) at (10.2,0.6) {$\neg\psi$};
     
      \path (a) edge[->] node[below] {$\neg\existsptrue{\pi}$} (b);
      \path (ep) edge[->] node {$\ast$} (fp);
      \path (e) edge[->] node {$+$} (ep);
      \path (fp) edge[->] node {$\neg\existsptrue{\pi}$} (f);

      \path[purple,thick,->,font=\footnotesize]
      (e) edge[bend left] node[above] {$\max \pi$} (ep);

      \path[purple,thick,->,font=\footnotesize]
      (f) edge[bend right] node[above] {$\max \pi$} (fp);
      
    \end{tikzpicture}
    \caption{Proof of Claim~\ref{claim:consecutive}(2.)\label{fig:consecutive2}}
    \end{figure}
    \end{enumerate}
   This concludes the proof of the claim. \qedhere(Claim~\ref{claim:consecutive}).
  \end{proof}

  To conclude the proof of Theorem~\ref{thm:PDLp-to-CFM}, consider the formulas
  $\varphi_1=\existsptrue{\pi}$, $\varphi_2=\Loop{\maxpi{\pi}}$, 
  $\varphi_3=\Loop{\minpi{({\rightp}\cdot\pi^{-1})}}$, and
  $\varphi_4=\Loop{\maxpi{((\maxpi{\pi})\cdot\rightg{\neg\existsptrue{\pi}})}}$.
  By Lemmas~\ref{lem:trad-loop-free} and~\ref{lem:det-loop}, we already
  have transducers $\A_{\varphi_i}$ for $i\in\{1,2,3,4\}$.
  We let $\A_\psi = \A \circ (\A_{\varphi_1} \times\A_{\varphi_2}
  \times\A_{\varphi_3}\times\A_{\varphi_4})$,
  where, at an event $f$ labeled $(b_1,b_2,b_3,b_4)$, the transducer $\A$ 
  outputs
    $1$ if $b_3=1$ or if $(b_1,b_2,b_3,b_4)=(1,0,0,0)$ and the output was
    $1$ at the last event $e$ on the same process satisfying $\varphi_1$ (to do 
    so, each process keeps in its state the output at the last event where 
    $b_1$ was $1$),    
    and $0$ otherwise.

  Consider the formula $\varphi'$ over $\Sigma \times \{0,1\}$ obtained
  from $\varphi$ by replacing $\psi$ by $\bigvee_{a \in \Sigma} (a,1)$,
  and all event formulas $a$, with $a \in \Sigma$, by $(a,0) \lor (a,1)$.
  It contains fewer $\Loopop$ operators than $\varphi$, so by induction
  hypothesis, we have a transducer $\Aphi {\varphi'}$ for $\varphi'$.
  We then let $\Aphi \varphi = \Aphi {\varphi'} \circ (\A_{\mathit{Id}} \times \Aphi \psi)$, 
  where $\A_{\mathit{Id}}$ is the transducer for the identity relation.
\end{proof}

 \begin{proof}[Proof of Proposition~\ref{prop:fo-cfm}]
   By Theorem~\ref{thm:FO-to-PDLm-main}, every $\FOle$ formula $\Phi(x)$ with a
   single free variable is equivalent to some $\PDLm$ state
   formula, for which we obtain a transducer $\A_\Phi$ using
   Theorem~\ref{thm:PDLp-to-CFM}. It is easy to build from $\A_\Phi$ CFMs
   for the sentences $\forall x. \Phi(x)$ and $\exists x. \Phi(x)$.
   Closure of $\LangCFM$ under union and intersection takes care of disjunction and conjunction.
 \end{proof}

% ------------------------------------------------------------
% ------------------------------------------------------------
% ------------------------------------------------------------

\section{Applications to Existentially Bounded MSCs}\label{sec:applications}

\newcommand\linle{\preceq}
\newcommand\linlt{\prec}
\newcommand\clinle{\preceq_B}
\newcommand\clinlt{\prec_B}
\newcommand\wlin[1]{M_{#1}}
\newcommand\Slin{\Sigma_{\mathit{lin}}}
\newcommand\type{\mathsf{type}}
\newcommand\LinB{\mathit{Lin}^B}
\newcommand\LM{L}
\newcommand\rev{\mathit{rev}}
\newcommand\revb[1]{\mathit{rev}_{#1}}
\newcommand\Llin{L_{\mathit{lin}}}
\newcommand\Philin{\Phi_{\mathit{lin}}}
\newcommand\uparrowB{\uparrow_B}
\newcommand\Pdiff[2]{P_{{\uparrowB {#1}} - {\uparrowB {#2}}}}
\newcommand\myPdiff[2]{\ploc(#1 - #2)}
\newcommand\myPmin[2]{\textup{minloc}(#1,#2)}

Though the translation of $\EMSO$/$\FO$ formulas into CFMs is interesting on its own, it allows us to obtain some difficult results for bounded CFMs as corollaries.

\subsection{Existentially bounded MSCs}

The first logical characterizations of communicating
finite-state machines were obtained for classes of \emph{bounded} MSCs.
Intuitively, this corresponds to restricting the channel capacity.  Bounded MSCs
are defined in terms of linearizations.  A \emph{linearization} of a given MSC
$M=(E,\prel,\mrel,\ploc,\lambda)$ is a total order ${\preceq} \subseteq E \times
E$ such that ${\le} \subseteq {\preceq}$.  For $B \in \N$, we call
$\preceq$ $B$-bounded if, for all $g \in E$ and $(p,q) \in \Ch$,
$|\{(e,f) \in {\mrel} \cap (E_p \times E_q) \mid e \preceq g \prec f\}| \le B$. In other words, the number of pending messages in $(p,q)$ never exceeds $B$.
There are (at least) two natural definitions of bounded MSCs: We call $M$
$\exists B$-bounded if $M$ has \emph{some} $B$-bounded linearization.
Accordingly, it is $\forall B$-bounded if \emph{all} its linearizations are
$B$-bounded.

\begin{example}
  The MSC from Figure~\ref{fig:msc} is $\exists 1$-bounded and $\forall
  4$-bounded.  These bounds are tight: the MSC is not $\forall 3$-bounded,
  because the four send events for, say, channel $(p_1,p_3)$ can be scheduled
  before the first reception $g_0$.

  As another example, consider the set of MSCs over two processes, $p$ and $q$,
  that consist of an arbitrary number of messages from $p$ to $q$ (and only
  messages from $p$ to $q$).  This language is $\exists 1$-bounded (every
  message may be received right after it was sent), but it is not $\forall
  B$-bounded, no matter what $B$ is.
\end{example}

In the following, we will consider only $\exists B$-bounded MSCs.
The set of $\exists B$-bounded MSCs is denoted by $\ebMSCs{\Procs}{\Sigma}{B}$.

Below, we show the following results.  First, for a given channel
bound $B$, the set $\ebMSCs{\Procs}{\Sigma}{B}$ is $\FO[\prel,\mrel,\le]$-definable (essentially due
to \cite{LohreyMuscholl04}).  By Theorem~\ref{thm:main}, we obtain
\cite[Proposition~5.14]{GKM06} stating that this set is recognized by some CFM.
Second, we obtain \cite[Proposition~5.3]{GKM06}, a Kleene theorem for
existentially bounded MSCs, as a corollary of Theorem~\ref{thm:main} in
combination with a linearization normal form from \cite{ThiagarajanW02}.

\medskip

Let $M=(E,\prel,\mrel,\ploc,\lambda)$ be some MSC, and
$e_1 \linlt e_2 \cdots \linlt e_n$ a linearization of $M$.
Given $e \in E$, we write $\type(e) = p$ if $e$ is an internal event on
process~$p$, $\type(e) = p!q$ if $e$ is a write on channel $(p,q)$, and
$\type(e) = q?p$ if $e$ is a read from channel $(p,q)$.
We associate with the linearization $\linle$ a word $\wlin \linle$
over the
alphabet $\Slin = \Sigma \times (P \cup \{q?p, p!q \mid (p,q) \in \Ch\})$.
More precisely, we let $\wlin \linle = a_1 \ldots a_n$ where
$a_i = (\lambda \times \type)(e_i)$.
Note that $M$ can be retrieved from $\wlin{\linle}$.
We let $\LinB(M) = \{\wlin \linle \mid \text{$\linle$ is a $B$-bounded linearization of $M$}\}$.

\begin{fact}[\!\!{\cite[Theorem~4.1]{GKM06}}]\label{EB-regular}
  Let $B \in \N$ and $L$ be a set of $\exists B$-bounded MSCs.
  The following are equivalent:
  \begin{enumerate}
  \item $L = \msclang(\A)$ for some CFM $\A$.
  \item $L = \msclang(\Phi)$ for some MSO formula $\Phi$.
  \item $\LinB(L)$ is a regular language.
  \end{enumerate}
\end{fact}

The proof given in \cite{GKM06} relies on the theory of Mazurkiewicz traces.
Another major part of the proof is the construction of a CFM recognizing the
set $\ebMSCs{\Procs}{\Sigma}{B}$ of $\exists B$-bounded MSCs
\cite[Proposition~5.14]{GKM06}.
We show that this CFM can in fact be obtained as a simple application of
Theorem~\ref{thm:main}.
Moreover, we give an alternative proof of $(3) \implies (1)$
(Section~5 in \cite{GKM06}).

\subsection{A CFM for Existentially Bounded MSCs}

The set $\ebMSCs{\Procs}{\Sigma}{B}$ of $\exists B$-bounded MSCs is in fact
$\FO[\mrel,\prel,\le]$-definable, and thus, we can apply Theorem~\ref{thm:main}
to construct a CFM $\A_{\exists B}$ recognizing $\ebMSCs{\Procs}{\Sigma}{B}$.
We describe below a formula defining $\ebMSCs{\Procs}{\Sigma}{B}$.

Let us first recall a characterization of $\exists B$-bounded MSCs.
Let $M=(E,\prel,\mrel,\ploc,\lambda)$ be an MSC.
We define a relation $\revb{B} \subseteq E \times E$ which consists of the set of
pairs $(f,g)$ such that $f$ is a receive event from some channel $(p,q)$ with
corresponding send event $e \mrel f$, and $g$ is the $B$-th send on
channel $(p,q)$ after event $e$.
The relation $\revb{B}$ is illustrated in Figure~\ref{fig:revb}
(represented by the dashed edges) for $B=1$
and an $\exists 1$-bounded MSC.
It can be defined by the $\sfPDLall[\cup]$ path formula
$$
\mathsf{rev}_B = \bigcup_{p\neq q} 
{\mrel_{p,q}^{-1}}\cdot \Big(\rightg{\neg\existsptrue{\mrel_{p,q}}}\cdot
\test{\existsptrue{\mrel_{p,q}}} \Big)^B\,.
$$
For completeness, let us also give a corresponding $\FO[\prel,\mrel,\le]$ formula:
\begin{multline*}
  \revb{B}(x,y) := \exists z_0, z_1, \ldots, z_B.\
  z_0 \mrel x \land z_B = y \land
  \bigwedge_{1 \le i \le B} \exists x_i.\ z_i \mrel x_i \land x \leproc x_i
   \\
  {}\land \bigwedge_{0 \le i < B-1} z_i \ltproc z_{i+1} \land
  \lnot (\exists z',x'.\ z_i \ltproc z' \ltproc z_{i+1} \land z' \mrel x' \land
  x \leproc x') \, .
\end{multline*}

\begin{fact}[\!\!\cite{LohreyMuscholl04}]
  $M$ is $\exists B$-bounded iff the relation $({<} \cup \revb{B})$ is acyclic.
\end{fact}

Note that, if $({<} \cup \revb{B})$ contains a cycle, then it contains one of
size at most $2|P|$. So $M$ is $\exists B$-bounded iff it satisfies the
$\sfPDLall[\Loopname,\cup]$ formula 
$\pdlsentence_{\exists B} = \neg\E\Loop{\mathsf{lt}_B}$ where
$$
  \mathsf{lt}_B = \bigcup_{2\leq n\leq |P|}
  \big( (\mrel \cup \mathsf{rev}_B)\cdot {\rightp} \big)^{n}
  \qquad\qquad \mrel = \bigcup_{p\neq q}\mrel_{p,q}\,.
$$
Again, let us determine a corresponding $\FO[\prel,\mrel,\le]$ formula:
$$
  \Phi_{\exists B} =
  \bigwedge_{2 \le n \le 2|P|} \lnot \Big(
  \exists x_0, \ldots, x_{n}.\ x_0 = x_{n} \wedge
  \bigwedge_{0 \le i < n}
    x_i < x_{i+1} \lor \revb{B}(x_i,x_{i+1}) \Big) \, .
$$

\newcommand{\convexpath}[2]{
  [   
  create hullcoords/.code={
    \global\edef\namelist{#1}
    \foreach [count=\counter] \nodename in \namelist {
      \global\edef\numberofnodes{\counter}
      \coordinate (hullcoord\counter) at (\nodename);
    }
    \coordinate (hullcoord0) at (hullcoord\numberofnodes);
    \pgfmathtruncatemacro\lastnumber{\numberofnodes+1}
    \coordinate (hullcoord\lastnumber) at (hullcoord1);
  },
  create hullcoords
  ]
  ($(hullcoord1)!#2!-90:(hullcoord0)$)
  \foreach [
  evaluate=\currentnode as \previousnode using \currentnode-1,
  evaluate=\currentnode as \nextnode using \currentnode+1
  ] \currentnode in {1,...,\numberofnodes} {
    let \p1 = ($(hullcoord\currentnode) - (hullcoord\previousnode)$),
    \n1 = {atan2(\y1,\x1) + 90},
    \p2 = ($(hullcoord\nextnode) - (hullcoord\currentnode)$),
    \n2 = {atan2(\y2,\x2) + 90},
    \n{delta} = {Mod(\n2-\n1,360) - 360}
    in 
    {arc [start angle=\n1, delta angle=\n{delta}, radius=#2]}
    -- ($(hullcoord\nextnode)!#2!-90:(hullcoord\currentnode)$) 
  }
}

\definecolor{mygreen}{RGB}{0.0,180,0.0}

\begin{figure}[h]
\centering
    \begin{tikzpicture}[semithick,>=stealth]

     \node[acirc] (p1) at (0,4) {};
     \node[acirc] (p2) at (2,4) {};
     \node[acirc] (p3) at (4,4) {};
     \node[acirc] (p4) at (6,4) {};
     \draw[->] (p1) -- (p2);\draw[->] (p2) -- (p3);\draw[->] (p3) -- (p4);

     \node[acirc,label=left:$e$] (q1) at (0,3) {};
     \node[acirc] (q2) at (2,3) {};
     \node[acirc] (q3) at (3,3) {};
     \node[acirc] (q4) at (4,3) {};
     \node[acirc] (q5) at (5,3) {};
     \node[acirc] (q6) at (6,3) {};
     \draw[->] (q1) -- (q2);\draw[->] (q2) -- (q3);\draw[->] (q3) -- (q4);\draw[->] (q4) -- (q5);\draw[->] (q5) -- (q6);

     \node[acirc,label=left:$f$] (r1) at (0,2) {};
     \node[acirc] (r2) at (2,2) {};
     \node[acirc] (r3) at (3,2) {};
     \node[acirc] (r4) at (4,2) {};
     \node[acirc] (r5) at (5,2) {};
     \draw[->] (r1) -- (r2);\draw[->] (r2) -- (r3);\draw[->] (r3) -- (r4);\draw[->] (r4) -- (r5);%\draw[->] (r5) -- (r6);
     
     \node[acirc] (s1) at (0,1) {};
     \node[acirc] (s2) at (2,1) {};
     \node[acirc] (s3) at (4,1) {};
     \draw[->] (s1) -- (s2);\draw[->] (s2) -- (s3);

     \draw[->] (p1) -- (q1);\draw[->] (p2) -- (q2);\draw[->] (p3) -- (q4);\draw[->] (p4) -- (q6);
     \draw[->] (s1) -- (r1);\draw[->] (s2) -- (r2);\draw[->] (s3) -- (r4);
     \draw[->] (q3) -- (r3);\draw[->] (q5) -- (r5);
    
    \draw[->,dotted] (q1) -- (p2);\draw[->,dotted] (q2) -- (p3);\draw[->,dotted] (q4) -- (p4);
    \draw[->,dotted] (r1) -- (s2);\draw[->,dotted] (r2) -- (s3);
    \draw[->,dotted] (r3) -- (q5);

    \begin{pgfonlayer}{background}
    {\fill[mygreen!30] \convexpath{q1,p2,p4,q6,r5,r3,q3}{3mm}; }
    \end{pgfonlayer}
 
    \begin{pgfonlayer}{background}
    {\fill[purple!30,opacity=0.6] \convexpath{r1,r3,q5,q6,r5,s3,s2}{3mm}; }
    \end{pgfonlayer}
    
      \node at (-1,4) {$p_1$};
      \node at (-1,3) {$p_2$};
      \node at (-1,2) {$p_3$};
      \node at (-1,1) {$p_4$};

      \node at (6.6,3.6) {${\uparrowB}{e}$};
      \node at (5.1,1.2) {${\uparrowB}{f}$};
                        
    \end{tikzpicture}
    \caption{The relation $\revb{B}$ for $B = 1$, and the sets ${\uparrowB}{e}$ and ${\uparrowB}{f}$\label{fig:revb}}
  \end{figure}
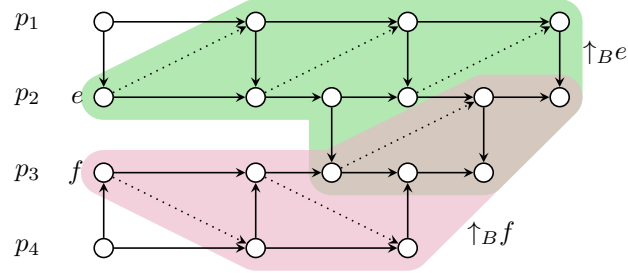

\newcommand\newPdiff[2]{\ploc({\uparrowB}#1 \setminus {\uparrowB}#2)}

% =======

\subsection{FO-definable Linearizations for Existentially Bounded MSCs}

We give a canonical $B$-bounded linearization of $\exists B$-bounded
MSCs, adapted from \cite[Definition~13]{ThiagarajanW02} where the
definition was given for traces.
We fix some total order $\sqsubseteq$ on $P$.
Let $M=(E,\prel,\mrel,\ploc,\lambda)$ be an $\exists B$-bounded MSC,
and let ${\le}_B = ({\le} \cup \revb{B})^\ast$ which is a partial order on $M$.
Note that a linearization of $M$ is $B$-bounded iff it contains $\le_B$.

For $e \in E$, we define ${\uparrowB}e = \{g \in E \mid e \le_B g\}$.
Moreover, for $E' \subseteq E$, let $\ploc(E') = \{\ploc(e) \mid e \in E'\}$.
Finally, given $e,f \in E$, let $e \parallel_B f$ if $e \not\le_B f$ and $f \not\le_B e$.
We then define a relation ${\clinlt} \subseteq E \times E$ by
$$\begin{array}{rcl}
e \clinlt f & \Longleftrightarrow~
\left(
\begin{array}{rl}
& e <_B f\\
\vee & e \parallel_B f ~\wedge~ \min (\newPdiff e f) \sqsubset \min (\newPdiff f e)
\end{array}
\right).
\end{array}
$$

\begin{example}
Consider the MSC $M$ in Figure~\ref{fig:revb} and suppose $p_1 \sqsubset p_2 \sqsubset p_3 \sqsubset p_4$.
We have $\newPdiff e f = \{p_1,p_2\}$ and $\newPdiff f e = \{p_3,p_4\}$.
Since $p_1 \sqsubset p_3$, we obtain $e \clinlt f$.
\end{example}

The following result is due to \cite[Lemma~14]{ThiagarajanW02}.
It is stated there for traces, but the proof can be taken almost verbatim.

\begin{lemma}\label{lem:linord}
The relation $\clinlt$ is a strict linear order on $E$.
\end{lemma}

Moreover, the reflexive closure
$\clinle$ of $\clinlt$
contains $\le_B$, hence it is a $B$-bounded linearization of $M$.
Finally, the relation $\clinlt$ is $\FO[\prel,\mrel,\le]$-definable.
Indeed, the strict partial order $<_B$ is 
$\FO[\prel,\mrel,\le]$-definable since it can be expressed with the path formula
$\mathsf{lt}_B$ given above.
From its definition, we deduce that the relation $\clinlt$ is also
$\FO[\prel,\mrel,\le]$-definable.

\medskip

We are now ready to give our alternative proof of the direction
$(3) \implies (1)$ in Fact~\ref{EB-regular}.

\begin{proof}[Proof of $\boldsymbol{(3) \implies (1)}$]
Let $L$ be a set of $\exists B$-bounded MSCs such that $\LinB(L)$ is regular.
There exists an EMSO sentence $\Philin$ over $\Slin$-labeled words such that $\LinB(L) = \msclang(\Philin)$.
Since $\clinle$ is $\FO[\prel,\mrel,\le]$-definable, it is easy to translate
$\Philin$ into an $\EMSO[\prel,\mrel,\le]$ formula $\Phi$ such that, for all
$\exists B$-bounded MSC $M$, we have $M \models \Phi$ iff
$\wlin {\clinle} \models \Philin$.
Let $\A$ be a CFM such that $\msclang(\A) = \msclang(\Phi \land \Phi_{\exists B})$.
Then, for all $M \in L$, $M$ is $\exists B$-bounded and
$\wlin {\clinle} \models \Philin$, hence $M\models\Phi \land \Phi_{\exists B}$, 
i.e., $M \in \msclang(\A)$.
Conversely, if $M \in \msclang(\A)$, then $M$ is $\exists B$-bounded and $\clinle$ is a
linearization of $M$. Moreover, $\wlin {\clinle} \in \LinB(L)$, hence $M \in L$.
\end{proof}

% ------------------------------------------------------------
% ------------------------------------------------------------
% ------------------------------------------------------------
\section{Conclusion}\label{sec:concl}

In this paper, we showed that every $\FO[\prel,\mrel,\leq]$ formula over MSCs
is effectively equivalent to a CFM.
As an intermediate step, we used a purely logical transformation of own interest,
relating FO logic with a star-free fragment of PDL.

While star-free PDL constitutes a two-dimensional temporal logic  over MSCs,
we leave open whether there is a one-dimensional one, with a finite set of
$\FO$-definable modalities, that is expressively complete for $\FO[\prel,\mrel,\leq]$.

It will be worthwhile to see whether our techniques can be applied
to other settings such as trees or Mazurkiewicz traces.

\bibliography{lit}

\begin{thebibliography}{10}

\bibitem{BFG-stacs18}
B.~Bollig, M.~Fortin, and P.~Gastin.
\newblock Communicating finite-state machines and two-variable logic.
\newblock In {\em 35th Symposium on Theoretical Aspects of Computer Science
  (STACS 2018)}, volume~96 of {\em Leibniz International Proceedings in
  Informatics}, pages 17:1--17:14. Leibniz-Zentrum f{\"u}r Informatik, 2018.

\bibitem{BKM-lmcs10}
B.~Bollig, D.~Kuske, and I.~Meinecke.
\newblock Propositional dynamic logic for message-passing systems.
\newblock {\em Logical Methods in Computer Science}, 6(3:16), 2010.

\bibitem{BolligJournal}
B.~Bollig and M.~Leucker.
\newblock Message-passing automata are expressively equivalent to {EMSO} logic.
\newblock {\em Theoretical Computer Science}, 358(2-3):150--172, 2006.

\bibitem{Brand1983}
D.~Brand and P.~Zafiropulo.
\newblock On communicating finite-state machines.
\newblock {\em Journal of the ACM}, 30(2), 1983.

\bibitem{Buechi:60}
J.~B{\"u}chi.
\newblock Weak second order logic and finite automata.
\newblock {\em Z. Math. Logik, Grundlag. Math.}, 5:66--62, 1960.

\bibitem{GiacomoL94}
G.~{De Giacomo} and M.~Lenzerini.
\newblock Boosting the correspondence between description logics and
  propositional dynamic logics.
\newblock In {\em Proceedings of the 12th National Conference on Artificial
  Intelligence, Seattle, WA, USA, July 31 - August 4, 1994, Volume 1.}, pages
  205--212. {AAAI} Press / The {MIT} Press, 1994.

\bibitem{DiGa08Thomas}
V.~Diekert and P.~Gastin.
\newblock First-order definable languages.
\newblock In J{\"o}rg Flum, Erich Gr{\"a}del, and Thomas Wilke, editors, {\em
  Logic and Automata: History and Perspectives}, volume~2 of {\em Texts in
  Logic and Games}, pages 261--306. Amsterdam University Press, 2008.

\bibitem{DiekertRozenberg95}
V.~Diekert and G.~Rozenberg, editors.
\newblock {\em The Book of Traces}.
\newblock World Scientific, Singapore, 1995.

\bibitem{Elgot1961}
C.~C. Elgot.
\newblock Decision problems of finite automata design and related arithmetics.
\newblock {\em Transactions of the American Mathematical Society}, 98:21--52,
  1961.

\bibitem{FisL79}
M.~J. Fischer and R.~E. Ladner.
\newblock Propositional {D}ynamic {L}ogic of regular programs.
\newblock {\em Journal of Computer and System Sciences}, 18(2):194--211, 1979.

\bibitem{Gabbay1981}
D.~M. Gabbay.
\newblock Expressive functional completeness in tense logic.
\newblock In Uwe M{\"o}nnich, editor, {\em Aspects of Philosophical Logic: Some
  Logical Forays into Central Notions of Linguistics and Philosophy}, pages
  91--117. Springer Netherlands, Dordrecht, 1981.

\bibitem{gabbay1994temporal}
D.~M. Gabbay, I.~Hodkinson, and M.~A. Reynolds.
\newblock {\em {Temporal Logic: Mathematical Foundations and Computational
  Aspects, vol. 1}}.
\newblock Oxford University Press, 1994.

\bibitem{GKM06}
B.~Genest, D.~Kuske, and A.~Muscholl.
\newblock {A {Kleene} theorem and model checking algorithms for existentially
  bounded communicating automata}.
\newblock {\em Information and Computation}, 204(6):920--956, 2006.

\bibitem{GKM07}
B.~Genest, D.~Kuske, and A.~Muscholl.
\newblock On communicating automata with bounded channels.
\newblock {\em Fundamenta Informaticae}, 80(1-3):147--167, 2007.

\bibitem{Goeller2009}
S.~G{\"o}ller, M.~Lohrey, and C.~Lutz.
\newblock {PDL} with intersection and converse: satisfiability and
  infinite-state model checking.
\newblock {\em Journal of Symbolic Logic}, 74(1):279--314, 2009.

\bibitem{GradelO99}
E.~Gr{\"{a}}del and M.~Otto.
\newblock On logics with two variables.
\newblock {\em Theoretical Computer Science}, 224(1-2):73--113, 1999.

\bibitem{HalpernM92}
J.~Y. Halpern and Y.~Moses.
\newblock A guide to completeness and complexity for modal logics of knowledge
  and belief.
\newblock {\em Artif. Intell.}, 54(2):319--379, 1992.

\bibitem{Hanf1965}
W.~Hanf.
\newblock Model-theoretic methods in the study of elementary logic.
\newblock In J.~W. Addison, L.~Henkin, and A.~Tarski, editors, {\em The Theory
  of Models}. North-Holland, Amsterdam, 1965.

\bibitem{HenriksenJournal}
J.~G. Henriksen, M.~Mukund, K.~Narayan Kumar, M.~Sohoni, and P.~S. Thiagarajan.
\newblock A theory of regular {MSC} languages.
\newblock {\em Information and Computation}, 202(1):1--38, 2005.

\bibitem{Kamp68}
H.~Kamp.
\newblock {\em Tense Logic and the Theory of Linear Order}.
\newblock PhD thesis, University of California, Los Angeles, 1968.

\bibitem{Kuske01}
D.~Kuske.
\newblock {Regular sets of infinite message sequence charts}.
\newblock {\em Information and Computation}, 187:80--109, 2003.

\bibitem{Lamport78}
L.~Lamport.
\newblock Time, clocks, and the ordering of events in a distributed system.
\newblock {\em Commun. {ACM}}, 21(7):558--565, 1978.

\bibitem{Lange06}
M.~Lange.
\newblock Model checking propositional dynamic logic with all extras.
\newblock {\em Journal of Applied Logic}, 4(1):39--49, 2006.

\bibitem{LangeLutzJSL05}
M.~Lange and C.~Lutz.
\newblock {2-ExpTime lower bounds for Propositional Dynamic Logics with
  intersection}.
\newblock {\em Journal of Symbolic Logic}, 70(5):1072--1086, 2005.

\bibitem{LohreyMuscholl04}
M.~Lohrey and A.~Muscholl.
\newblock Bounded {MSC} {C}ommunication.
\newblock {\em Information and Computation}, 189(2):160--181, 2004.

\bibitem{Marx05}
M.~Marx.
\newblock {Conditional XPath}.
\newblock {\em {ACM} Trans. Database Syst.}, 30(4):929--959, 2005.

\bibitem{Mennicke13}
R.~Mennicke.
\newblock Propositional dynamic logic with converse and repeat for
  message-passing systems.
\newblock {\em Logical Methods in Computer Science}, 9(2:12):1--35, 2013.

\bibitem{phd-stockmeyer}
L.~J. Stockmeyer.
\newblock {\em The Complexity of Decision Problems in Automata Theory and
  Logic}.
\newblock PhD thesis, MIT, 1974.

\bibitem{Streett81}
R.~S. Streett.
\newblock Propositional dynamic logic of looping and converse.
\newblock In {\em Proceedings of STOC'81}, pages 375--383. {ACM}, 1981.

\bibitem{ThiagarajanW02}
P.~S. Thiagarajan and I.~Walukiewicz.
\newblock An expressively complete linear time temporal logic for
  {Mazurkiewicz} traces.
\newblock {\em Inf. Comput.}, 179(2):230--249, 2002.

\bibitem{Tho97handbook}
W.~Thomas.
\newblock Languages, automata and logic.
\newblock In A.~Salomaa and G.~Rozenberg, editors, {\em Handbook of Formal
  Languages}, volume~3, pages 389--455. Springer, 1997.

\bibitem{Trakhtenbrot62}
B.~A. Trakhtenbrot.
\newblock Finite automata and monadic second order logic.
\newblock {\em Siberian Math. J}, 3:103--131, 1962.
\newblock In Russian; English translation in {\sl Amer. Math. Soc. Transl.} 59,
  1966, 23--55.

\end{thebibliography}

% --------------------------------------------------------
% --------------------------------------------------------
% --------------------------------------------------------

\end{document}